\documentclass[10pt,draftcls,onecolumn]{IEEEtran} 
\usepackage{flushend}

\usepackage{balance}

\usepackage[caption=false,font=footnotesize]{subfig}
\usepackage{amssymb}
\usepackage{amsthm}
\usepackage{latexsym}
\usepackage{verbatim}
\usepackage[pdftex]{graphicx}
\usepackage{fixltx2e}
\usepackage{array}
\usepackage{booktabs}
\usepackage[square, numbers,sort&compress]{natbib}
\usepackage[cmex10]{amsmath}
\usepackage{setspace,algorithm,algorithmic,enumerate}

\newcommand{\ComplexityFont}[1]{%
{\ensuremath{\complexity@possiblymakesmaller{\complexity@fontcommand{#1}}}}
}

\renewcommand{\ComplexityFont}[1]{{\ensuremath{\mathsf{#1}}}}

\newtheorem{theorem}{Theorem}

\newtheorem{lemma}{Lemma}

\newtheorem{definition}{Definition}

\newtheorem{problem}{Problem}

{\begin{list}               
    {$\bullet$ \hfill}{
        \setlength{\leftmargin}{\parindent}
        \setlength{\parsep}{0.04\baselineskip}
        \setlength{\itemsep}{0.5\parsep}
        \setlength{\labelwidth}{\leftmargin}
        \setlength{\labelsep}{0em}}
    }
{\end{list}}

\providecommand{\cref}[1]{Chapter~\ref{chap:#1}}

\providecommand{\R}{\ensuremath{\mathbb{R}}}
\providecommand{\C}{\ensuremath{\mathbb{C}}}

\providecommand{\MSE}[1]{\operatorname{MSE}\left(#1\right)}

\providecommand{\set}[1]{\left\{#1\right\}}

\renewcommand{\vec}[1]{\ensuremath{\boldsymbol{#1}}}
\providecommand{\mat}[1]{\ensuremath{\boldsymbol{#1}}}
\providecommand{\wh}[1]{\ensuremath{\widehat{#1}}}

\providecommand{\calS}{\mathcal{S}}

\providecommand{\calL}{\mathcal{L}}
\providecommand{\calX}{\mathcal{X}}

 \providecommand{\mT}{\mat{T}}

\providecommand{\mPsi}{{\mat{\Psi}}}
\providecommand{\mpsi}{{\mat{\psi}}}

 \providecommand{\vf}{\vec{f}}

\newcommand{\figdir}{./figures}

\begin{document}

\title{Near-Optimal Sensor Placement \\for Linear Inverse Problems}

\author{Juri Ranieri$^*$,~\IEEEmembership{Student Member,~IEEE}, Amina
  Chebira$^\dagger$,~\IEEEmembership{Member,~IEEE} 
  and~Martin~Vetterli$^*$,~\IEEEmembership{Fellow,~IEEE}
  \thanks{Copyright (c) 2013 IEEE. Personal use of this
      material is permitted. However, permission to use this material
      for any other purposes must be obtained from the IEEE by sending
      a request to pubs-permissions@ieee.org. $^*$ School of Computer and Communication Sciences, Ecole
    Polytechnique F\'ed\'erale de Lausanne (EPFL), CH-1015 Lausanne,
    Switzerland (contact e-mail juri.ranieri@epfl.ch).}
  \thanks{$^\dagger$ Swiss Center for Electronics and Microtechnology
    (CSEM), CH-2002 Neuch\^atel, Switzerland.}
 }

\maketitle

\begin{abstract} 
  A classic problem is the estimation of a set of parameters from
  measurements collected by only a few sensors. The number of sensors
  is often limited by physical or economical constraints and their
  placement is of fundamental importance to obtain accurate
  estimates. Unfortunately, the selection of the optimal sensor
  locations is intrinsically combinatorial and the available
  approximation algorithms are not guaranteed to generate good
  solutions in all cases of interest.

  We propose FrameSense, a greedy algorithm for the selection of
  optimal sensor locations. The core cost function of the algorithm is
  the \emph{frame potential}, a scalar property of matrices that
  measures the orthogonality of its rows. Notably, FrameSense is the
  first algorithm that is near-optimal in terms of mean square error,
  meaning that its solution is always guaranteed to be close to the
  optimal one.

  Moreover, we show with an extensive set of numerical experiments
  that FrameSense achieves state-of-the-art performance while having
  the lowest computational cost, when compared to other greedy
  methods.
\end{abstract}

\begin{IEEEkeywords}
Sensor placement, inverse problem, frame potential, greedy algorithm.
\end{IEEEkeywords}

\IEEEpeerreviewmaketitle

\section{Introduction}
\label{sec:introduction}

In many contexts, it is of interest to measure physical phenomena that
vary in space and time. Common examples are temperature, sound, and
pollution. Modern approaches tackling this problem are often based on
wireless sensor networks (WSN), namely systems composed of many
sensing nodes, each capable of measuring, processing and communicating
information about the surrounding environment.

Challenges and trade-offs characterize the design of a WSN. One of the
key aspects to design a successful WSN is the optimization of the
spatial locations of the sensors nodes, given the location's impact on
many relevant indicators, such as coverage, energy consumption and
connectivity. When the data collected by the WSN is used to solve
inverse problems, the optimization of the sensor locations becomes
even more critical. In fact, the location of the sensor nodes
determines the error of the solution of the inverse problem and its
optimization represents the difference between being able to obtain a
reasonable solution or not. In this work we consider linear inverse
problems defined as
\begin{align}
  \vf=\mPsi\vec{\alpha},
\label{eq:lin_model}
\end{align}
where $\vf\in\R^N$ is the measured physical field,
$\vec{\alpha}\in\R^K$ are the parameters to be estimated and
$\mPsi\in\R^{N\times K}$ is the known linear model representing the
relationship between the measurements and the parameters. Note that
this simple model can be easily adapted to more complicated
scenarios. For example, if the collected measurements are linear
combinations of the physical field, as in the presence of a sampling
kernel, we simply consider $\mPsi=\mat{\Phi}\mat{\Theta}$, where
$\mat{\Phi}$ and $\mat{\Theta}$ represent the sampling kernel and the
physical phenomenon, respectively.

The role of $\vec{\alpha}$ depends on the specific inverse
problem. For example, if the WSN is designed for \emph{source
  localization}, $\vec{\alpha}$ represents the location and the
intensity of the field sources. On the other hand, if we are planning
to \emph{interpolate} the measured samples to recover the entire
field, we may think of $\vec{\alpha}$ as its low-dimensional
representation.  In other scientific applications, for example
\cite{Vetterli:2002bs,Baraniuk:2009ve, Bruckstein:2009il}, the
solution of a linear inverse problem is a step within a complex
procedure and $\vec{\alpha}$ may not have a direct
interpretation. Nonetheless, the accurate estimation of $\vec{\alpha}$
is of fundamental importance.

It is generally too expensive or even impossible to sense the physical
field $\vf$ with $N$ sensor nodes, where $N$ is determined by the
resolution of the discrete physical field. Assume we have only $L<N$
sensors, then we need to analyze how to choose the $L$ sampling
locations such that the solution of the linear inverse problem
\eqref{eq:lin_model} has the least amount of error. Namely, we would
like to choose the \emph{most informative} $L$ rows of $\mPsi$ out of
the $N$ available ones. One could simply adopt a brute force approach
and inspect all the possible combinations for the $L$ sensor
locations.  In this case, the operation count is exponential due to its
combinatorial nature, making the approach unfeasible even for modest
values of $N$.

It is possible to significantly reduce the computational cost by
accepting a sub-optimal sensor placement produced by an approximation
algorithm. In this case, near-optimal algorithms are desired since
they \emph{always} produce a solution of \emph{guaranteed quality}. We
measure this quality as the ratio between the value of the
approximated solution and the value of the optimal one, and we call it
the \emph{approximation factor}.

\subsection{Problem Statement and Prior Art}


We consider the linear model introduced in \eqref{eq:lin_model} and a
WSN measuring the field at only $L< N$ locations. We denote the sets of
measured locations and of available locations as
$\mathcal{L}=\{i_1,\ldots,i_L\}$ and $\mathcal{N}=\{1,\ldots,N\}$,
respectively. Note that $\mathcal{L}\subseteq \mathcal{N}$ and
$|\calL|=L$.

The measured field is denoted as $\vf_{\cal{L}}\in\R^{L}$, where the
subscript represents the selection of the elements of $\vf$ indexed by
$\cal{L}$. Consequently, we define a pruned matrix
$\mPsi_{\cal{L}}\in\R^{L\times K}$, where we kept only the rows of
$\mPsi$ indexed by $\cal{L}$. We obtain a smaller linear system of
equations,
\begin{align}
\vf_{\cal{L}}=\mPsi_{\cal{L}}\vec{\alpha},
\label{eq:red_lin_mod}
\end{align}
where we still recover $\vec{\alpha}$, but with a reduced set of
measurements, $L\ge K$. Note that we have by definition
$\mPsi_{\mathcal{N}}=\mPsi$ and $\vf_{\mathcal{N}}=\vf$.

Given the set of measurements $\vf_\mathcal{L}$, there may not exist
an $\widehat{\vec{\alpha}}$ that solves \eqref{eq:red_lin_mod}. If it
exists, the solution may not be unique. To overcome this problem, we
usually look for the least squares solution, defined as
$\widehat{\vec{\alpha}}=\arg\min_{\vec{\alpha}}\|\mPsi_{\cal{L}}\vec{\alpha}-\vf_{\cal{L}}\|_2^2.$
Assume that $\mPsi_{\cal{L}}$ has rank $K$, then this solution is
found using the Moore-Penrose pseudoinverse,
\begin{align}
\widehat{\vec{\alpha}}=\mPsi_{\cal{L}}^+\vf_{\cal{L}},\nonumber
\end{align}
where
$\mPsi_{\cal{L}}^+=(\mPsi_{\cal{L}}^*\mPsi_{\cal{L}})^{-1}\mPsi_{\cal{L}}^*$. The
pseudoinverse generalizes the concept of inverse matrix to non-square
matrices and is also known as the \emph{canonical dual frame} in frame
theory. For simplicity of notation, we introduce
$\mathbf{T}_\mathcal{L}=\mPsi_{\cal{L}}^*\mPsi_{\cal{L}}\in\R^{K\times
  K}$, a Hermitian-symmetric matrix that strongly influences the
reconstruction performance. More precisely, the error of the least
squares solution depends on the spectrum of $\mT_{\cal{L}}$. That is,
when the measurements $\vf_{\cal{L}}$ are perturbed by a zero-mean
i.i.d. Gaussian noise with variance $\sigma^2$, the mean square error
($\operatorname{MSE}$) of the least squares solution
\cite{Fickus:2011vq} is
\begin{align}
  \MSE{\wh{\vec{\alpha}}}=\|\wh{\vec{\alpha}}-\vec{\alpha}\|_2^2=\sigma^2\sum_{k=1}^K\frac{1}{\lambda_k},
\label{eq:MSE}
\end{align}
where $\lambda_k$ is the $k$-th eigenvalue of the matrix
$\mT_{\cal{L}}$. We thus state the sensor placement problem as
follows.
\begin{problem}
Given a matrix $\mPsi\in\R^{N\times K}$ and a number of sensors $L$,
find the sensor placement $\cal{L}$ such that
\begin{align}
 \arg\min_{\cal{L}} \sum_{k=1}^K\frac{1}{\lambda_k} && \text{subject
   to }&& |\mathcal{L}|=L.
\end{align}
\end{problem}
Note that if $\mT_\mathcal{L}$ is rank deficient, that is
$\operatorname{rank}(\mT_\calL)<K$, then the $\operatorname{MSE}$ is
not bounded.

A trivial choice would be to design algorithms minimizing directly the
$\operatorname{MSE}$ with some approximation procedure, such as greedy
ones. In practice, the $\operatorname{MSE}$ is not used because it has
many unfavorable local minima. Therefore, the research effort is
focused in finding tight proxies of the $\operatorname{MSE}$ that can
be efficiently optimized. In what follows, we survey different
approximation strategies and proxies from the literature.

\subsection{Prior work}
\label{sec:literature}

Classic solutions to the sensor placement problem can be classified
in three categories: convex optimization, greedy methods and
heuristics.

Convex optimization methods \cite{Shamaiah:2010hj,Joshi:2009el} are
based on the relaxation of the Boolean constraints $\{0,1\}^N$
representing the sensor placement to the convex set $[0,1]^N$. This
relaxation is usually not tight as heuristics are needed to choose the
sensor locations and there is no a-priori guarantee on the distance
from the optimal solution. The authors in \cite{Joshi:2009el} define
an online bound for the quality of the obtained solution by looking at
the gap between the primal and the dual problem.

Heuristic methods \cite{AlObaidy:2008wn,Mukherjee:2006joa,
  Lau:2008de,Chiu:2004wh,MacKay:1992ul,Wang:2004uf} are valid options
to reduce the cost of the exhaustive search, which has a prohibitive
cost. Again, even if the methods work in practice, little can be said
about the quality or the optimality of the solution.

Greedy algorithms leveraging the submodularity of the cost function
\cite{Nemhauser:1978vz} are a class of algorithms having polynomial
complexity and guaranteed performance with regards to the chosen cost
function \cite{Shamaiah:2010hj,Krause:2008vo,Das:2008uc,Naeem:2009th,
  Das:2011ue}. Since the MSE is not submodular in general
\cite{Das:2008uc,Das:2011ue}, alternative cost functions have been
considered
\cite{Shamaiah:2010hj,Krause:2008vo,Das:2008uc,Naeem:2009th,
  Das:2011ue}. The proposed methods are theoretically near-optimal
with regards to the chosen cost function, but little can be said about
the achieved MSE. Moreover, the local optimization of the proposed
cost functions are computationally demanding, often requiring the
inversion of large matrices \cite{Krause:2008vo}. Therefore,
approximations of the cost functions have been proposed
\cite{Krause:2008vo}, offering a significant speedup for an acceptable
reduction of the solution's quality.

Beside the approximation strategy, approximation algorithms are
differentiated by the chosen cost function. Under restrictive
assumptions, the $\operatorname{MSE}$ can be chosen as a cost
function, see \cite{Das:2008uc,Golovin:2010va}. In \cite{Das:2011ue},
the authors bounded the performance of greedy algorithms optimizing
$\operatorname{R}^2$, a measure of goodness of fit based on the
$\operatorname{MSE}$, using the concept of \emph{submodularity
  ratio}. However, such algorithm is generally less performant than
the greedy algorithms optimizing proxies of the
$\operatorname{MSE}$. Common proxies of the $\operatorname{MSE}$ are
inspired by \emph{information theoretic} measures such as entropy
\cite{Wang:2004uf}, cross-entropy
\cite{Ramakrishnan:2005wn,Naeem:2009th} and mutual information
\cite{Krause:2008vo}. A popular choice is the maximization of the log
determinant of $\mT_\mathcal{L}$, being the volume of the confidence
ellipsoid given by the measurements. This proxy has been historically
introduced in D-Optimal experiment design \cite{Steinberg:1984wg}, but
has also been successfully proposed as a cost function for a convex
relaxed method \cite{Joshi:2009el} and greedy algorithms
\cite{Shamaiah:2010hj}. Other proxies have also been introduced in
optimal experiment design, such as maximization of the smallest
eigenvalue $\lambda_K$ (E-Optimal design) or the maximization of the
trace of $\mT_\calL$ (T-Optimal design). A detailed description of the
different choices available for experiment design can be found in
\cite{Steinberg:1984wg}.

\begin{figure}[t!]  \centering
  \subfloat{\label{fig:a}\includegraphics[scale=0.47]{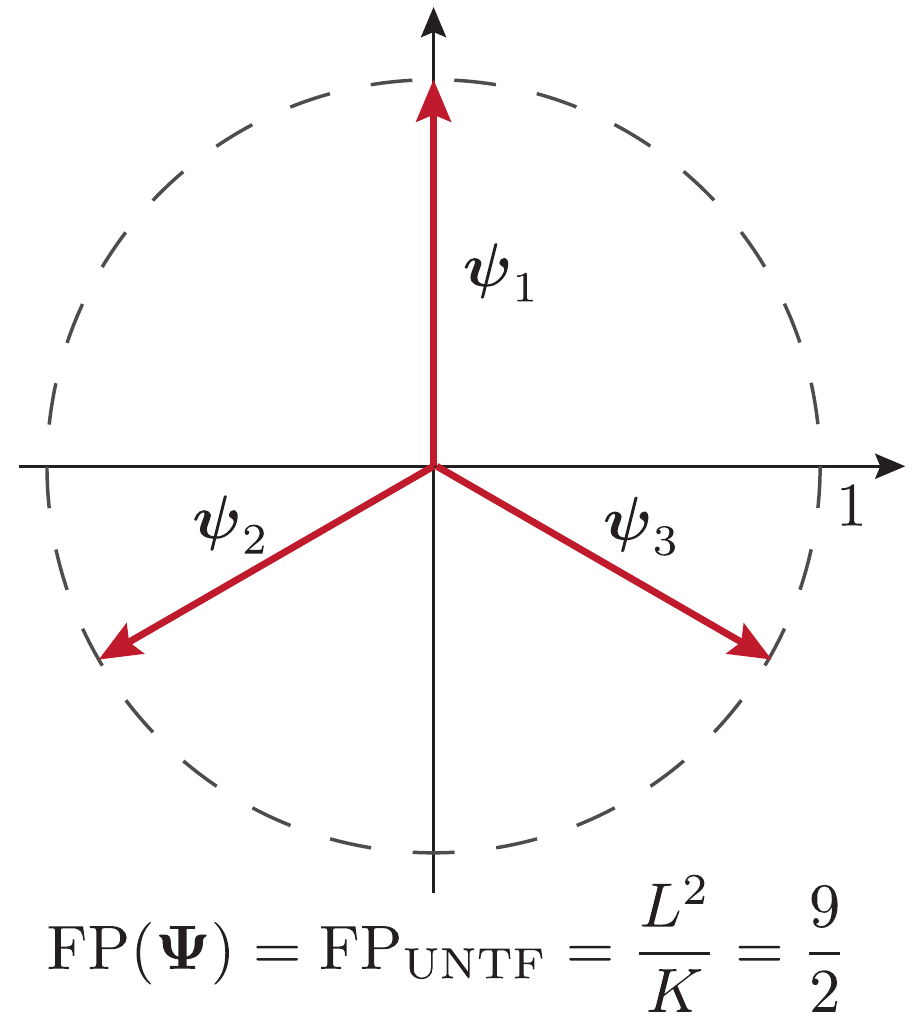}}
  \subfloat{\label{fig:b}\includegraphics[scale=0.47]{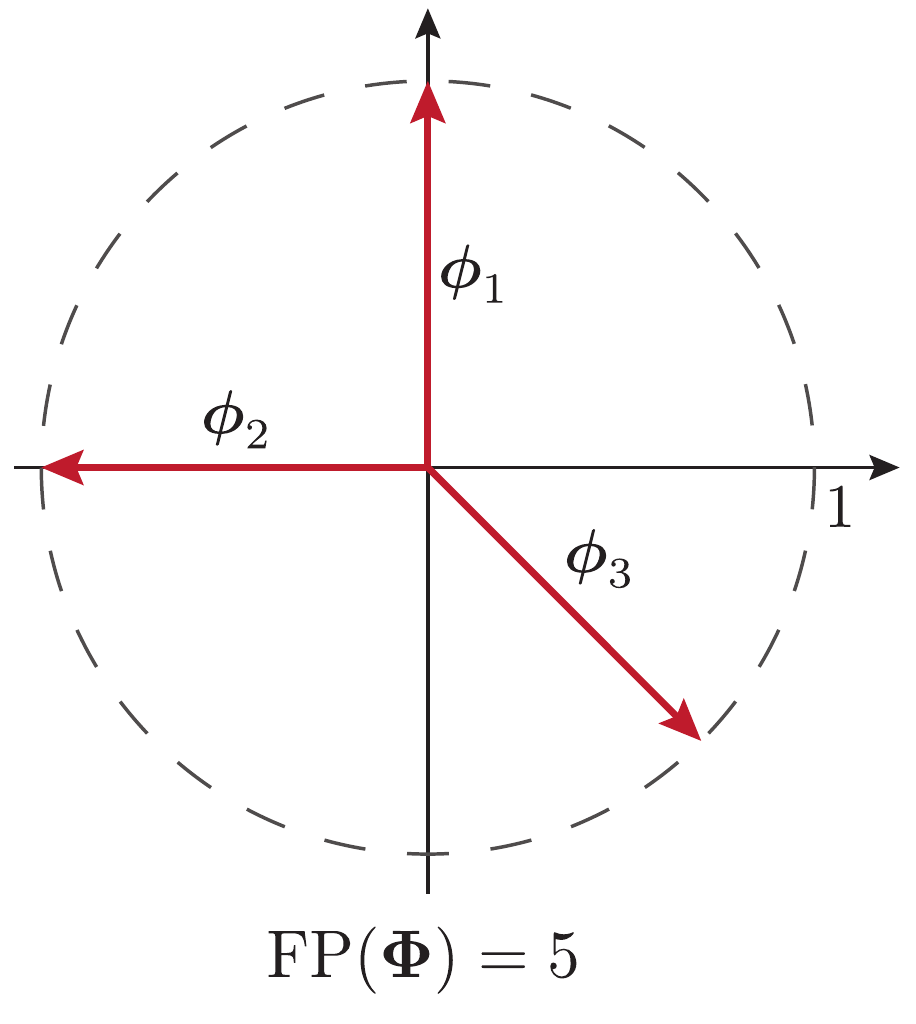}}
\vspace{-0.3cm}
\caption{A graphical representation of the unit-norm rows of two
  matrices $\mPsi$ and $\mat{\Phi}$. On the left, the rows $\mPsi$ are
  optimally spread on $\R^2$ minimizing the $\operatorname{FP}$ and
  the $\operatorname{MSE}$. This is also know as the Mercedes-Benz
  frame and is a typical example of a unit-norm tight frame. On the
  right, $\mPsi$ is the frame built by adding a vector to the
  orthonormal basis of $\R^2$ $[\mat{\phi_1},\mat{\phi}_2]$. The three
  vectors are not minimizing the $\operatorname{FP}$ and therefore
  they are not in equilibrium with regards to the frame force. You can
  envision a parallel example with three electrons on a unit circle
  under the Coulomb force. More precisely, if the three electrons are
  free to move, they would reach an equilibrium when located as the
  vectors of $\mPsi$, up to a rotation factor. }
  \label{fig:FP_ex}
\vspace{-0.4cm}
\end{figure}

Note that there exists optimal strategies with a reasonable
computational cost for some specific scenarios. This is the case when
we have the freedom of completely designing the matrix
$\mPsi_{\cal{L}}$ given the dimensions $L$ and $K$. More precisely, if
$L=K$, the optimal basis corresponds to an orthonormal basis, while if
$L>K$, then we are looking for a unit-norm tight frame
\cite{Goyal:1998tc,Benedetto:2003ut}. Benedetto et al. showed that
each tight frame is a non-unique global minimizer of the \emph{frame
  potential} ($\operatorname{FP}$), that is a scalar property of the
frame defined as
\begin{align}
\operatorname{FP}(\mPsi_{\cal{L}})=\sum_{i,j\in\cal{L}}|\left<\mpsi_i,\mpsi_j\right>|^2,\nonumber
\end{align}
where $\mpsi_i$ is the $i$th row of $\mPsi_{\cal{L}}$. One of the
reasons of the popularity of the $\operatorname{FP}$ in the frame
theory community is its interesting physical interpretation
\cite{Casazza:2006wl}. Namely, it is the potential energy of the
so-called \emph{frame force}, a force between vectors inspired by the
Coulomb force.  The frame force and its potential energy have been
introduced for their \emph{orthogonality encouraging} property: the
force is repulsive when the angle between the two vectors is acute,
null when they are orthogonal and attractive when the angle is
obtuse. A graphical explanation of this physical interpretation is
given in Figure \ref{fig:FP_ex}, where the unit-norm rows of two
matrices $\mPsi$ and $\mat{\Phi}$ belonging to $\R^{3\times 2}$ are
represented. While $\mPsi$, that is the unit-norm tight frame
minimizing the $\operatorname{FP}$, has vectors as close to
orthogonality as possible and therefore in equilibrium with regards to
the frame force, the vectors of $\mat{\Phi}$ are not optimally close
to orthogonality and the $\operatorname{FP}$ is thus not
minimized. Note that according to frame theory
\cite{Benedetto:2003ut}, $\mPsi$ is the matrix that also achieves the
minimum $\operatorname{MSE}$ (per component).

Note that the $\operatorname{FP}$ is also known as the \emph{total
  summed correlation} in communication theory \cite{Rupf:1994fl}. It
is used to optimize the signatures of CDMA systems to achieve the
Welch's lower bound \cite{Kovacevic:2007db} and maximize the capacity
of the channel.

Given its interpretation and its role in defining the existence of
tight frames---the optimal frames in terms of
$\operatorname{MSE}$---we hypothesize that the $\operatorname{FP}$ is
an interesting cost function for an approximation algorithm.

\subsection{Our contributions}

We propose FrameSense, a greedy sensor placement method that
minimizes the frame potential to choose the sensing locations
$\mathcal{L}$. We briefly summarize the innovative aspects of the
proposed algorithm:
\begin{itemize}
\item Under some stability conditions regarding the spectrum of
  $\mPsi$, FrameSense is the only known algorithm, to the best of our
  knowledge, that is near-optimal with regards to $\operatorname{MSE}$.
\item FrameSense outperforms other greedy algorithms in terms of
  $\operatorname{MSE}$.
\item FrameSense is on par with the method based on convex relaxation
  \cite{Joshi:2009el}, which uses heuristics to improve the local
  solution and has a significantly higher complexity.
\item The computational cost of FrameSense is significantly lower
  with regards to the other considered algorithms. 
\end{itemize}

\vspace{.5cm} The remainder of the paper uses the following notations:
calligraphic letters as $\mathcal{C}$ indicate sets of elements, while
bold letters as $\mPsi$ and $\mpsi$ indicates matrices and vectors,
respectively.  The $k$-th largest eigenvalue of
$\mathbf{T}_\mathcal{L}=\mPsi^*_\mathcal{L}\mPsi_\mathcal{L}$ is
denoted as $\lambda_k$. Moreover, we always consider a real physical
field $\vf$ and real matrices $\mPsi$ for simplicity but the extension
to the complex domain does not require major adjustments.
\vspace{.5cm}

The content is organized as follows. In Section \ref{sec:FP} we
introduce some frame theory concepts focusing on the role of the
$\operatorname{FP}$. We describe FrameSense and the analysis of its
near-optimality in Section \ref{sec:smartsense} and we numerically
compare its performance with various other algorithms in Section
\ref{sec:num_exp}.

\section{The frame potential in frame theory}
\label{sec:FP} 
This section briefly introduces some of the basic concepts of frame
theory that are useful to understand and analyze the proposed
algorithm.  Frame theory studies and designs families of matrices
$\mPsi_\mathcal{L}$ such that $\mT_\mathcal{L}$ is
well-conditioned. More precisely, $\mPsi_\mathcal{L}$ is a frame for a
Hilbert space $\mathbb{H}$ if there exists two scalars $A$ and $B$
such that $0<A\le B<\infty$ so that for every
$\mathbf{x}\in\mathbb{H}$ we have
\begin{align}
A\|\mathbf{x}\|_2^2\le \|\mPsi_\mathcal{L} \mathbf{x}\|_2^2 \le
B\|\mathbf{x}\|_2^2, \nonumber
\end{align}
where $A$ and $B$ are called frame bounds. $\mPsi_\mathcal{L}$ is a
tight frame when $A=B$ and its columns are orthogonal by
construction. Of particular interest is the case of unit norm tight
frames (UNTF); these are tight frames whose frame elements---the rows
$\mpsi_i$---have unit norm. These provide Parseval-like relationships,
despite the non-orthogonality of the frame elements of
$\mPsi_\mathcal{L}$. In addition to \eqref{eq:MSE}, there are other
interesting relationships between the characteristics of
$\mPsi_\mathcal{L}$ and the spectrum of $\mT_\mathcal{L}$. For
example, we can express the $\operatorname{FP}$ as
\begin{align}
\operatorname{FP}(\mPsi_\mathcal{L})=\operatorname{Trace}({\mT_\calL}^*\mT_\calL)=\sum_{k=1}^K
|\lambda_k|^2. \nonumber
\end{align} 
Moreover, the sum of the eigenvalues of $\mT_\mathcal{L}$ is equal to
the sum of the norm of the rows, $\sum_{i=1}^{L} \|\mpsi_i\|^2 =
\sum_{k=1}^K \lambda_k$. 

These quantities of interest take a simplified analytical form for
UNTFs. In this scenario, we know \cite{Benedetto:2003ut} that the
$\operatorname{FP}$ is minimum with regards to all other matrices of
the same size with unit-norm rows, and it is equal to
$\operatorname{FP}_\text{UNTF}=\frac{L^2}{K}$. According to
\cite{Benedetto:2003ut}, the optimal $\operatorname{MSE}$ is also
achieved when the $\operatorname{FP}$ is minimized and it is equal to
$\operatorname{MSE}_\text{UNTF}=\frac{K^2}{L}$. Note that in this case
all the eigenvalues are equal,
$\lambda_\text{UNTF}=\lambda_i=\frac{L}{K}\;\forall i$.

Next, we would like to intuitively explain why the $\operatorname{FP}$
is a good candidate to be a proxy for the
$\operatorname{MSE}$. Consider the distance between the
$\operatorname{FP}$ of a matrix with unit-norm rows
$\mPsi_\calL\in\R^{L\times K}$ and the $\operatorname{FP}$ of a
UNTF. Then, it is possible to show that this distance
is always positive and equal to
\begin{align}
\operatorname{FP}(\mPsi_\calL)-\operatorname{FP}_\text{UNTF}=\sum_{k=1}^{K}\left(\lambda_k-\frac{L}{K}\right)^2,\nonumber
\end{align}
where $\frac{L}{K}$ is the value of the eigenvalues
$\lambda_\text{UNTF}$. Note that if we minimize the
$\operatorname{FP}$ of $\mPsi_\calL$, then each $\lambda_k$ converges
to $\frac{L}{K}$.

At the same time, the distance between the $\operatorname{MSE}$ of
$\mPsi_\calL$ and the $\operatorname{MSE}$ of the UNTF can be
expressed as
\begin{align}
\operatorname{MSE}(\mPsi_\calL)-\operatorname{MSE}_\text{UNTF}=\sum_{k=1}^K\left(\frac{1}{\lambda_k}-\frac{K}{L}\right).\nonumber
\end{align}
Now, it is easy to see that if the eigenvalues converge to
$\frac{L}{K}$, then $\operatorname{MSE}(\mPsi_\calL)$ converges to the
$\operatorname{MSE}$ of a UNTF, being also the optimal one
\cite{Benedetto:2003ut}.

\section{FrameSense: a near-optimal sensor placement}
\label{sec:smartsense}
Even if the intuition given in Section \ref{sec:FP} is clear, it does
not directly explain why an algorithm placing the sensor according to
the $\operatorname{FP}$ would perform well in terms of
$\operatorname{MSE}$.  Indeed, we need to address some complications
such as matrices $\mPsi$ having rows with different norms and the
non-uniform convergence of the eigenvalues. In what follows, we first
describe the details of FrameSense and then analyze its
near-optimality in terms of both the $\operatorname{FP}$ and the
$\operatorname{MSE}$.

\subsection{The algorithm}
FrameSense finds the sensor locations $\mathcal{L}$ given the known
model $\mPsi$ and the number of available sensors nodes $L$ with a
greedy minimization of the $\operatorname{FP}$. It is a greedy
``worst-out'' algorithm: at each iteration it removes the row of
$\mPsi$ that maximally increases the $\operatorname{FP}$. In other
words, we define a set of locations $\calS$ that are not suitable for
sensing and at each iteration we add to $\calS$ the row that maximizes
the following cost function:
\begin{align}
  F(\mathcal{S})=\operatorname{FP}(\mPsi)-\operatorname{FP}(\mPsi_{\mathcal{N}\setminus\mathcal{S}}).
\label{eq:cost_fun}
\end{align}
The pseudo-code for FrameSense is given in Algorithm
\ref{alg:greedyFP}.

\begin{algorithm} 
  \algorithmicrequire{~Linear Model $\mPsi$, Number of sensors
    $L$} \\
  \algorithmicensure{~ Sensor locations $\cal{L}$}\\
  \noindent\rule[1ex]{240pt}{.3pt}
  \begin{enumerate}
\item Initialize the set of locations, $\mathcal{L}=\emptyset$.
\item Initialize the set of available locations, ${\cal{N}}=\{1,\ldots,N\}$.
\item Find the first two rows to eliminate, $\mathcal{S}=\arg
  \max_{i,j\in\mathcal{N}}|\left<\mpsi_{i},\mpsi_{j}\right>|^2$.
\item Update the available locations, $\mathcal{L}=\mathcal{N}\setminus\mathcal{S}$.

\item {\bf Repeat until $L$ locations are found}
  \begin{enumerate}
  \item If $|\mathcal{S}|=N-L$, stop.
  \item Find the optimal row, $i^*=\arg \min_{i\in{\cal{L}}}
   F\left({\cal{S}}\cup i\right)$.
  \item Update the set of removed locations, $\mathcal{S}=\mathcal{S}\cup i^*$. 
\item Update the available locations,
  $\mathcal{L}=\mathcal{L}\setminus i^*$.
\end{enumerate}
\end{enumerate}
  \caption{FrameSense}
  \label{alg:greedyFP}
\end{algorithm}

One may ask why we do not optimize directly the $\operatorname{MSE}$,
instead of minimizing the $\operatorname{FP}$, which indirectly
optimizes the $\operatorname{MSE}$. As we have already indicated, a
greedy algorithm optimizing a general function, like the
$\operatorname{MSE}$, converges to a local stationary point of the
cost function and we have no guarantee on the distance from the global
optimum. On the other hand, we can prove that FrameSense is
near-optimal with regards to the $\operatorname{FP}$ by exploiting the
submodularity of the cost function. In addition, we also guarantee the
performance of FrameSense in terms of the $\operatorname{MSE}$,
exploiting a link between $\operatorname{FP}$ and
$\operatorname{MSE}$.

\subsection{Near-optimality of FrameSense with regards to $\operatorname{FP}$}
We define the performance of FrameSense with regards to
$\operatorname{FP}$ using the theory of submodular functions. We start
by defining the concept of submodularity that relates to the concept
of diminishing returns: if we add an element to a set $\mathcal{Y}$,
the benefit is smaller or equal than adding the same element to one of
the subsets of $\mathcal{Y}$. Then, we introduce a theorem by
Nemhauser et al. \cite{Nemhauser:1978vz} that defines the
approximation factor of greedy algorithms maximizing a submodular
function. We continue by showing that FrameSense satisfies the
conditions of Nemhauser's theorem and we derive its approximation
factor in terms of $\operatorname{FP}$.

\begin{definition}[Submodular function]
Given two sets
$\mathcal{X}$ and $\mathcal{Y}$ such that
$\mathcal{X}\subset\mathcal{Y}\subset\mathcal{N}$ and given an
element $i\in \mathcal{N}\setminus\mathcal{Y}$, a function $G$ is 
submodular if it satisfies
\begin{align}
G(\mathcal{X}\cup i)-G(\mathcal{X})\ge G(\mathcal{Y}\cup
i)-G(\mathcal{Y}).
\end{align}
\label{def:submodularity}
\end{definition}

Submodular functions are useful in combinatorial optimization because
greedy algorithms have favorable properties when optimizing a function
with such a property. More precisely, it has been proved that the greedy
maximization of submodular functions is near-optimal
\cite{Nemhauser:1978vz}.

\begin{theorem}[Near-optimal maximization of submodular function \cite{Nemhauser:1978vz}]
  Let $G$ be a normalized, monotone, submodular set function over a
  finite set $\mathcal{N}$. Let $\mathcal{L}$ be the set of $L$
  elements chosen by the greedy algorithm, and let $\operatorname{OPT}
  = \max_{\mathcal{A}\subset \mathcal{N},|\mathcal{A}|=L}
  G(\mathcal{A})$ be the optimal set of elements. Then
\begin{align}
G(\mathcal{L})\ge \left(1 -
  \frac{1}{e}\right)\operatorname{G}(\operatorname{OPT}), \nonumber
\end{align}
where $e$ is Euler's number. 
\label{thm:nemhauser}
\end{theorem}

Namely, if $G$ satisfies the conditions of Theorem
\ref{thm:nemhauser}, then the solution of the greedy algorithm is
always close to the optimal one. These conditions are satisfied by the
cost function $F$ in \eqref{eq:cost_fun}, as shown in the following
lemma.

\begin{lemma}[Submodularity of the cost function]
  The set function maximized in Algorithm \ref{alg:greedyFP},
\begin{align}
 F(\mathcal{S})=\operatorname{FP}(\mPsi)-\operatorname{FP}(\mPsi_{\mathcal{N}\setminus\mathcal{S}}),
\end{align}
 is a normalized, monotone, submodular function.
\label{lemma:submodular}
\end{lemma}
\begin{proof} The set function $F$ is normalized if
  $F(\emptyset)=0$. Here, normalization is trivially shown since
  $\mPsi=\mPsi_{\mathcal{N}}$ by definition. To show monotonicity, we
  pick a generic matrix $\mPsi$ of $N$ rows, a set $\mathcal{X}$ and
  an index $i\notin \mathcal{X}$. Then, we compute the increment of
  $F$ due to $i$ with regards to the set $\calX$, showing that it
  is always positive.
\begin{align}
F(\mathcal{X}\cup i)-F(\mathcal{X})&=\operatorname{FP}(\mPsi_{\mathcal{N}\setminus\mathcal{X}})-
\operatorname{FP}(\mPsi_{\mathcal{N}\setminus\mathcal{X}\cup{i}}) \nonumber \\
&\overset{(a)}{=}\hspace{-3mm}\sum_{n,m\in
  \mathcal{A}\cup{i}}\hspace{-2mm}|\langle\mpsi_n,\mpsi_m\rangle|^2-\hspace{-2mm}\sum_{n,m\in
  \mathcal{A}}\hspace{-1mm}|\langle\mpsi_n,\mpsi_m\rangle|^2\nonumber \\
&=2\sum_{n\in
  \mathcal{A}}|\langle\mpsi_n,\mpsi_i\rangle|^2+|\langle\mpsi_i,\mpsi_i\rangle|^2\ge
0, \nonumber
\end{align}
where $(a)$ is due to a change of variable
$\mathcal{N}\setminus\mathcal{X}=\mathcal{A}$.  Assuming without loss
of generality that $\mathcal{Y}=\mathcal{X}\cup j$, we check the
submodularity according to Definition \ref{def:submodularity} .
\begin{align}
&F(\mathcal{X}\cup i)-F(\mathcal{X})-F(\mathcal{Y}\cup
i)+F(\mathcal{Y})\nonumber\\
&=F(\mathcal{X}\cup i)-F(\mathcal{X})-F(\mathcal{X}\cup
\set{i,j})+F(\mathcal{X}\cup j)\nonumber \\
&=\operatorname{FP}(\mPsi_{\mathcal{A}\cup \set{i,j}})-
\operatorname{FP}(\mPsi_{\mathcal{A}\cup {j}}) - 
\operatorname{FP}(\mPsi_{\mathcal{A}\cup{i}})+
\operatorname{FP}(\mPsi_{\mathcal{A}}) \nonumber \\
&=2\sum_{n\in\mathcal{A}\cup{j}}|\langle \mpsi_n,\mpsi_i\rangle|^2 -
2\sum_{n\in\mathcal{A}}|\langle \mpsi_n,\mpsi_i\rangle|^2\nonumber \\
&=2|\langle\mpsi_i,\mpsi_j\rangle|^2\ge 0. \nonumber
\end{align}
\end{proof}
Now, we use Theorem \ref{thm:nemhauser} to derive the approximation
factor of FrameSense with regards to the $\operatorname{FP}$.

\begin{theorem}[$\operatorname{FP}$ approximation factor]
\label{thm:FP}
Consider a matrix $\mPsi\in\R^{N\times K}$ and a given number of
sensors $L$, such that $K\le L < N$. Denote the optimal set of
locations as $\text{OPT}=\arg\max_{\mathcal{A}\subset\mathcal{N},
  |\mathcal{A}|=L} \operatorname{FP}(\mPsi_\mathcal{A})$ and the
greedy solution found by FrameSense as $\mathcal{L}$. Then,
$\mathcal{L}$ is near-optimal with regards to the $\operatorname{FP}$,
\begin{align}
  \operatorname{FP}(\mPsi_\mathcal{L})\le \gamma \operatorname{FP}(\mPsi_\text{OPT}),
\end{align}
where
$\gamma=\left(1+\frac{1}{e}\left(\operatorname{FP}(\mPsi)\frac{K}{{{L_\text{MIN}}^2}}-1\right)\right)$
is the approximation factor and $L_\text{MIN}=\min_{|\mathcal{L}|=L}
\sum_{i\in\calL}\|\mpsi_i\|^2$ is the sum of the norms of the $L$ rows
with the smallest norm.  
\end{theorem}
\begin{IEEEproof}
  According to Lemma \ref{lemma:submodular}, the cost function used in
  FrameSense satisfies the conditions of Theorem
  \ref{thm:nemhauser}. Therefore,
\begin{align}
F(\overline{\text{OPT}})-F(\mathcal{S})\le \frac{1}{e}(F(\overline{\text{OPT}})),\nonumber
\end{align}
where
$F(\mathcal{S})=\operatorname{FP}(\mPsi)-\operatorname{FP}(\mPsi_{\mathcal{N}\setminus\mathcal{S}})$
is the considered cost function, $\mathcal{S}$ is the set of rows
eliminated by FrameSense and
$\overline{\text{OPT}}=\mathcal{N}\setminus \text{OPT}$. If we
consider the cost function, we obtain
\begin{align}
  \operatorname{FP}(\mPsi_{\mathcal{N}\setminus\mathcal{S}})\le
  \left(1-\frac{1}{e}\right)
  \operatorname{FP}(\mPsi_{\mathcal{N}\setminus\overline{\text{OPT}}})+
  \frac{1}{e}\operatorname{FP}(\mPsi).
\label{eq:boh2}
\end{align}
Then, we note that the following minimization problem,
\begin{align}
  \operatorname{minimize}_\mathcal{L}
  \operatorname{FP}(\mPsi_\mathcal{L})&&\text{subject to }&& |\mathcal{L}|=L, \nonumber
\end{align}
is equivalent, under the change of variable
$\mathcal{L}=\mathcal{N}\setminus\mathcal{S}$, to
\begin{align}
  \operatorname{minimize}_\mathcal{S}
  \operatorname{FP}(\mPsi_{\mathcal{N}\setminus\mathcal{S}})  &&
  \text{subject to }&&|\mathcal{S}|=N-L. \nonumber
\end{align}
Using the equivalence  in \eqref{eq:boh2}, we obtain
\begin{align}
\operatorname{FP}(\mPsi_\mathcal{L})\le\left(1
    +\frac{1}{e}\left(\frac{\operatorname{FP}(\mPsi)}{\operatorname{FP}(\mPsi_\text{OPT})}
      -1\right)\right)
  \operatorname{FP}(\mPsi_\text{OPT}), \nonumber
\end{align}
To conclude the proof, we bound from above the term
$\frac{1}{e}\frac{\operatorname{FP}(\mPsi)}{\operatorname{FP}(\mPsi_\text{OPT})}$. First,
  we consider the optimal solution $\text{OPT}$ to select a
  tight frame whose rows have a summed norm of
  $L_\text{OPT}=\sum_{i\in\text{OPT}}\|\mpsi_i\|^2$,
\begin{align}
  \operatorname{FP}(\mPsi_\mathcal{L})\le\left(1 +\frac{1}{e}
    \left(\operatorname{FP}(\mPsi) \frac{K}{{L_\text{OPT}}^2}
      -1\right)\right) \operatorname{FP}(\mPsi_\text{OPT}).\nonumber
\end{align}
Then, we assume that $\text{OPT}$ selects the rows having the smallest norm, $L_\text{OPT}\ge
L_\text{MIN}=\min_{|\mathcal{L}|=L}\sum_{i\in\calL}\|\mpsi_i\|^2.$
\end{IEEEproof}

Note that the $\operatorname{FP}$ of the original matrix influences
significantly the final result: the lower the $\operatorname{FP}$ of
$\mPsi$, the tighter the approximation obtained by the greedy
algorithm. Therefore, FrameSense performs better when the original
matrix $\mPsi$ is \emph{closer} to a tight frame. In fact, the
$\operatorname{FP}$ is bounded as follows, 
\begin{align}
\frac{\left(\sum_{i=1}^N\|\mpsi_i\|^2\right)^2}{K}\le\operatorname{FP}(\mPsi)\le \left(\sum_{i=1}^N\|\mpsi_i\|^2\right)^2\nonumber,
\end{align}
where the lower bound is reached by tight frames and the upper bound
by rank 1 matrices. 

Moreover, Theorem \ref{thm:FP} suggests to remove from $\mPsi$ the
rows whose norm is significantly smaller with regards to the others to
improve the performance of FrameSense. This suggestion is intuitive,
since such rows are also the least \emph{informative}.

\subsection{Near-optimality of FrameSense with regards to $\operatorname{MSE}$}
Having a near-optimal $\operatorname{FP}$ does not necessarily mean
that the obtained $\operatorname{MSE}$ is also near-optimal. Here, we
show that, under some assumptions on the spectrum of $\mPsi$,
FrameSense is near-optimal with regards to the $\operatorname{MSE}.$

Before going to the technical details, we generalize the concept of
number of sensors to account for the norms of the rows of
$\mPsi$. More precisely, we keep $L$ as the number of rows and we
define $L_\mathcal{A}=\sum_{i\in\mathcal{A}}\|\mpsi_i\|^2$ as the sum
of the norms of the rows of $\mPsi_\mathcal{A}$ for a generic set
$\mathcal{A}$. We also define the two extremal values of
$L_\mathcal{A}$,
\begin{align}
&L_\text{MIN}=\min_{\mathcal{A}\in\mathcal{N}, |\mathcal{A}|=L}
\sum_{i\in\mathcal{A}}\|\mpsi_i\|^2 \label{eq:Lmin},\\
&L_\text{MAX}=\max_{\mathcal{A}\in\mathcal{N}, |\mathcal{A}|=L}
\sum_{i\in\mathcal{A}}\|\mpsi_i\|^2,\label{eq:Lmax}
\end{align}
indicating respectively the minimum and the maximum value of
$L_\mathcal{A}$ among all possible selections of $L$ out of $N$ rows
of $\mPsi$. $L_\mathcal{A}$ is also connected to the spectrum of
$\mT_\mathcal{A}$. Indeed, $L_\mathcal{A}$ is the trace of
$\mT_\mathcal{A}$ and thus it is also the sum of its eigenvalues,
\begin{align}
L_\mathcal{A}=\sum_{i\in\mathcal{A}}\|\mpsi_i\|^2=\operatorname{Trace}(\mT_\mathcal{A})=\sum_{i\in\mathcal{A}}\lambda_i.\nonumber
\end{align}
If $\mPsi$ has rows with unit-norm, then
$L_\mathcal{A}=L_\text{MIN}=L_\text{MAX}=L$.

As a first step to prove the near-optimality
with regards to $\operatorname{MSE}$, we consider a possible placement
$\mathcal{A}$ and we bound the $\operatorname{MSE}$ of the matrix
$\mPsi_\mathcal{A}$ using its $\operatorname{FP}$ and the spectrum of
$\mT_\mathcal{A}$. To obtain such a bound, we use a known inequality
\cite{Sharma:2008ur} involving variance, arithmetic mean and harmonic
mean of a set of positive bounded numbers, in this case the
eigenvalues of $\mT_\mathcal{A}$. The following lemma describes the
bound, while its proof is given in Appendix \ref{app:bound_MSE}.

\begin{lemma}[MSE bound]
  Consider any $\mPsi_\mathcal{A}\in\R^{L\times K}$ with $|\mathcal{A}
  |=L\le K$ and denote the spectrum of $\mT_\mathcal{A}$ as
  $\lambda_1\ge\ldots\ge\lambda_K$.  Then the $\operatorname{MSE}$ is
  bounded according to the $\operatorname{FP}$ as,
\begin{align}
  &\operatorname{MSE}(\mPsi_\mathcal{A}) \le \frac {K}{L_\text{MIN}}
  \frac{\operatorname{FP}(\mPsi_\mathcal{A})}{{\lambda_K}^2}, \label{eq:MSEmax}\\
  &\operatorname{MSE}(\mPsi_\mathcal{A})\ge \frac {K}{L_\text{MAX}}
  \frac{\operatorname{FP}(\mPsi_\mathcal{A})}{{\lambda_1}^2},\label{eq:MSEmin}
\end{align}
where $L_\text{MIN}$ and $L_\text{MAX}$ are defined as in
\eqref{eq:Lmin} and \eqref{eq:Lmax}.
\label{lemma:MSE}
\end{lemma}

Lemma \ref{lemma:MSE} is key to study the approximation factor with
regards to the $\operatorname{MSE}$. Specifically, it allows to
analyze the two extremal cases:
\begin{itemize}
\item Given the optimal $\operatorname{FP}$, what is the lowest $\operatorname{MSE}$ we
  can achieve?
\item Given the worst case $\operatorname{FP}$ according to Theorem \ref{thm:FP}, what is
  the largest $\operatorname{MSE}$ we may encounter?
\end{itemize}

Lemma \ref{lemma:MSE} implies the necessity to properly bound the
spectrum of any $\mPsi_\mathcal{A}$ with a given $\operatorname{FP}$
obtained from $\mPsi$. While it is possible to bound $\lambda_1$ with
the $\operatorname{FP}$, it is also easy to build matrices with
$\lambda_K=0$, compromising the bound given in \eqref{eq:MSEmax}.
Therefore, we introduce the following property to control the
eigenvalues of any $\mT_\mathcal{A}$.

\begin{definition}[$(\delta,L)$-bounded frame] Consider a matrix
  $\mPsi\in\R^{N\times K}$ where $N\ge L$ and $N> K$. Then, we say
  that $\mPsi$ is $(\delta,L)$-bounded if, for every
  $\mathcal{A}\subseteq \mathcal{N}$ such that $|\mathcal{A}|=L$,
  $\mT_\mathcal{A}$ has a bounded spectrum
\begin{align}
  \frac{L_\text{MEAN}}{K}-\delta\le \lambda_i\le
  \frac{L_\text{MEAN}}{K}+\delta, \nonumber
\end{align}
where $1\le i\le K$, $\delta\ge 0$ and
$L_\text{MEAN}=\frac{L}{N}\sum_{i\in\mathcal{N}}\|\mpsi_i\|^2$ is
average value of $L_\mathcal{A}$.
\label{def:bounded_frame}
\end{definition}

The concept of $(\delta,L)$-bounded frames is similar to the notion of
matrices satisfying the restricted isometry property, that are used in
compressive sensing to guarantee the reconstruction of a sparse vector
from a limited number of linear measurements
\cite{Candes:2008pi}. Moreover, it allows us to define an
approximation factor for the $\operatorname{MSE}$ that does not depend
on the $\operatorname{FP}$, the cost-function we minimize.

\begin{figure*}[t!]  \centering
  \subfloat{\label{fig:a}\includegraphics[scale=0.450]{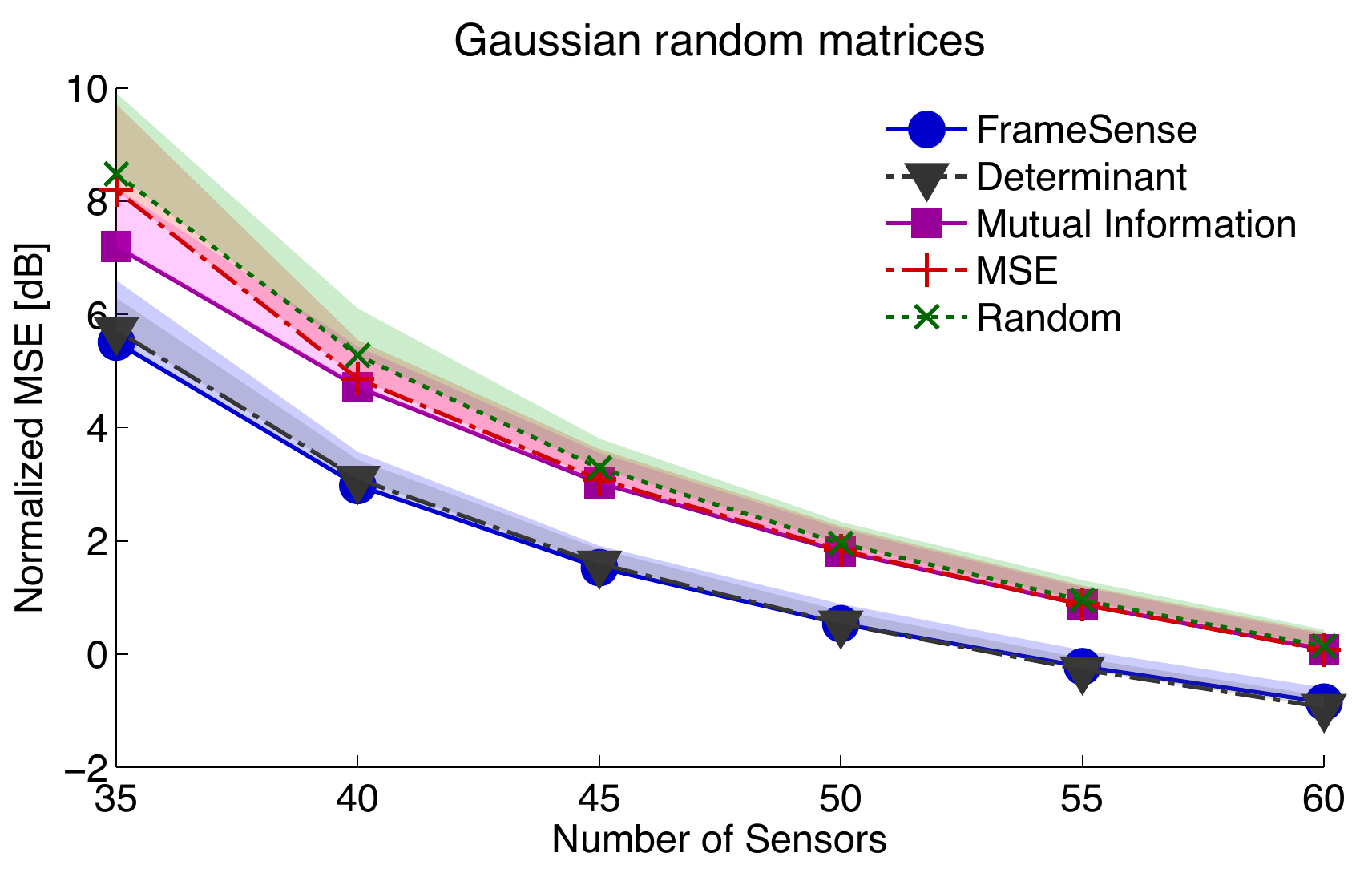}}
  \qquad
  \subfloat{\label{fig:b}\includegraphics[scale=0.45]{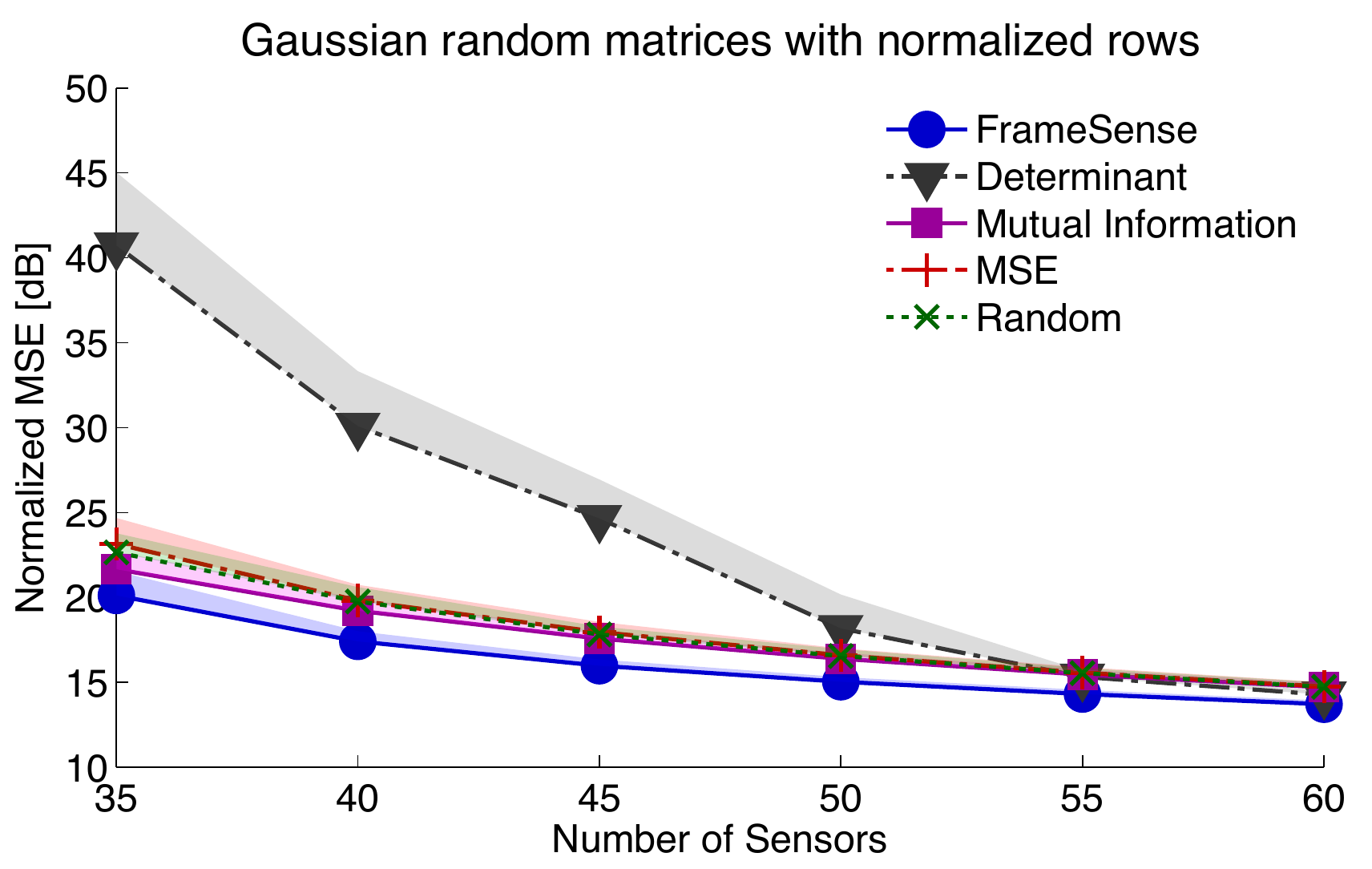}}\\
  \subfloat{\label{fig:c}\includegraphics[scale=0.45]{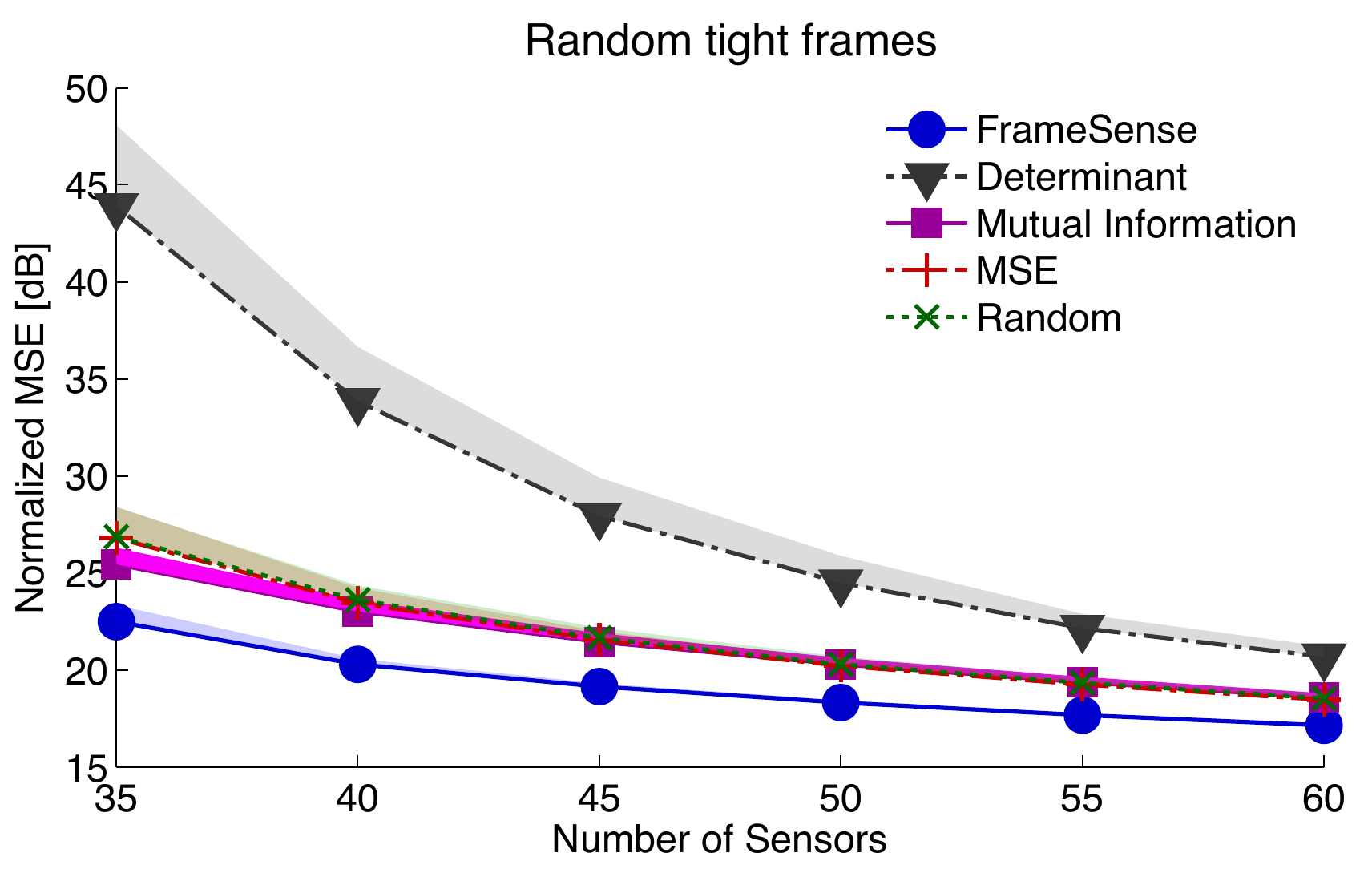}}
\qquad
  \subfloat{\label{fig:b}\includegraphics[scale=0.45]{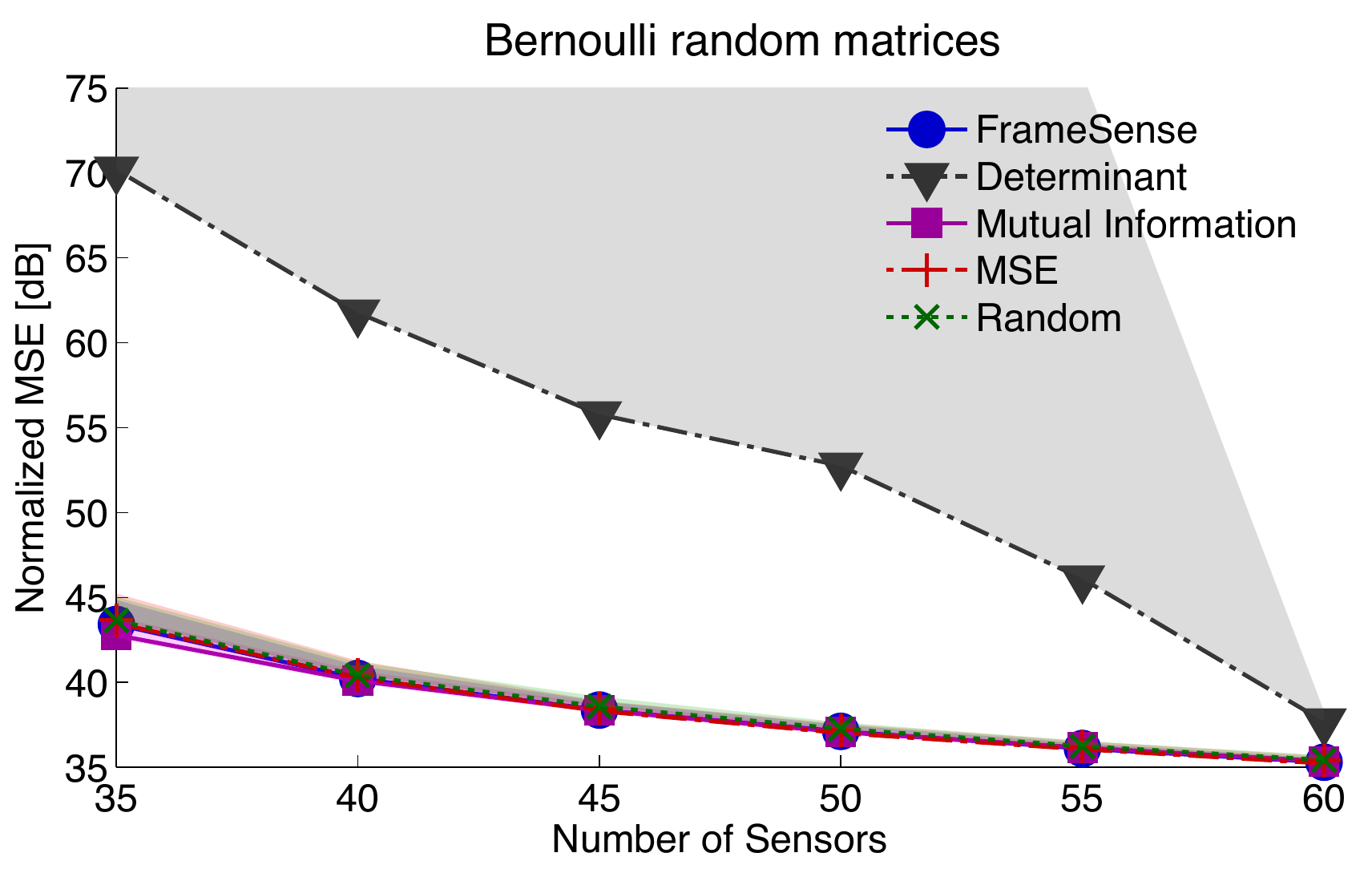}}
\vspace{-0.3cm}
\caption{Performance comparison between FrameSense and other greedy
  algorithms using common cost functions. We randomly
  generate matrices with $N=100$ and $K=30$ and test different greedy
  algorithms for a varying number of placed sensors $L$. The
  performance is measured in terms of $\operatorname{MSE}$, so the
  lower the curve, the higher the performance. The shaded areas
  represent the positive side of the error bars, measured using the
  standard deviation over 100 realizations. We consider four different
  types of sensing matrices, and in all cases FrameSense outperforms
  the other algorithms. We underline the consistency of FrameSense
  over the four types of matrices. }
  \label{fig:greedy_comp}
\vspace{-0.4cm}
\end{figure*}

\begin{figure}[t!]  \centering
  \includegraphics[scale=0.45]{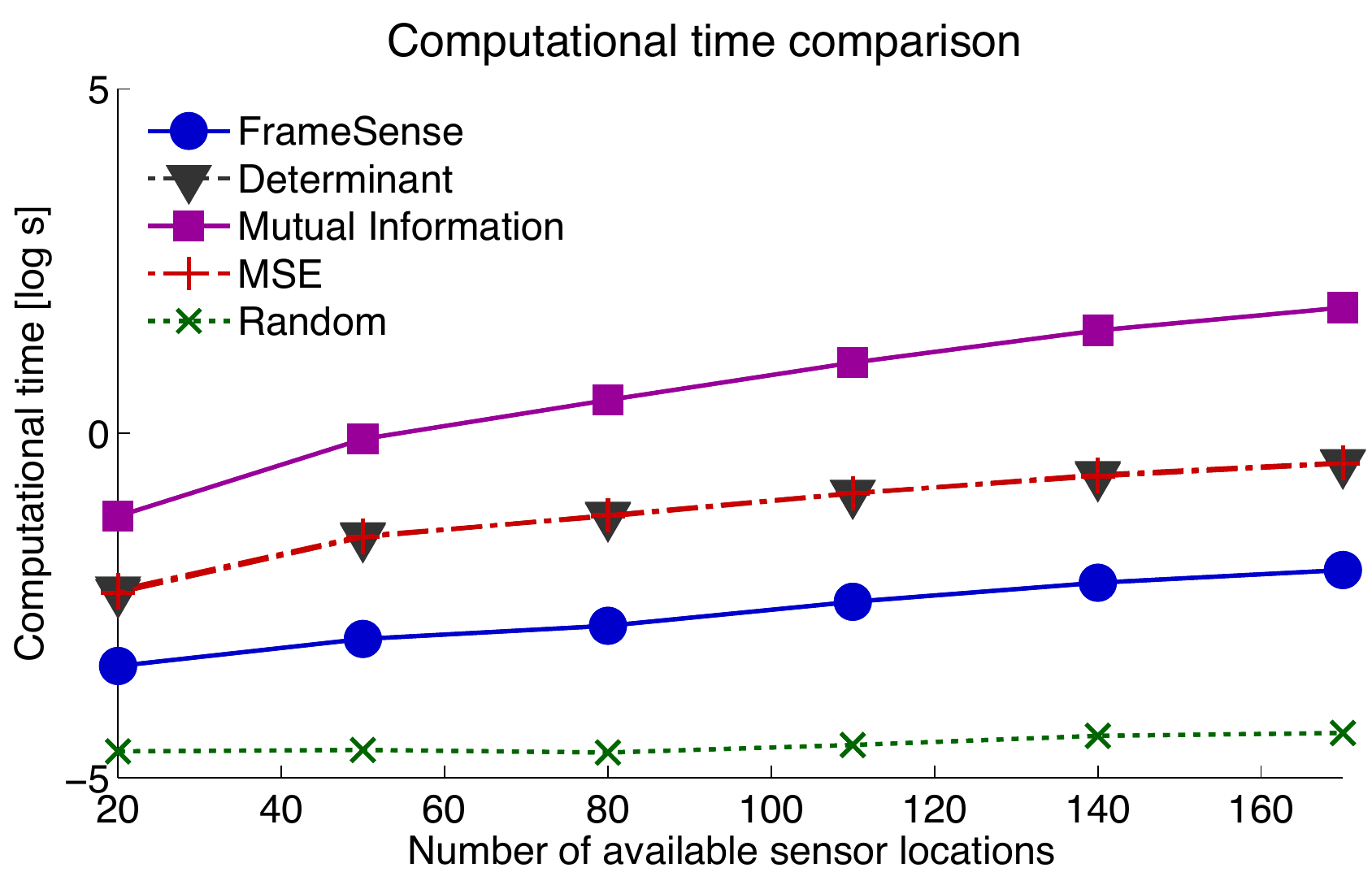}
\vspace{-0.3cm}
\caption{Comparison in terms of computational time as a function of
  $N$ between FrameSense and other greedy algorithms using commonly
  considered cost functions. The lower the curve the faster the
  algorithm. Note how FrameSense is the fastest algorithm, the
  exception being the random selection. Moreover, observe how the
  different algorithms scales equivalently with $N$. Errors bars are
  not shown because they are smaller than the markers.  }
  \label{fig:greedy_comp_time}
\vspace{-0.4cm}
\end{figure}

\newcommand{\slfrac}[2]{\left.#1\middle/#2\right.}

\begin{theorem}[MSE approximation factor for $(\delta,L)$-bounded frames]
  Consider a matrix $\mPsi\in\R^{N\times K}$ and $L\ge K$
  sensors. Assume $\mPsi$ to be a $(\delta,L)$-bounded frame, let $d$
  be the ratio $\slfrac{L_\text{MEAN}}{K}$ and define the optimal
  placement in terms of $\operatorname{MSE}$ as $\text{OPT}=\arg
  \min_{\mathcal{A}\in\mathcal{N},|\mathcal{A}|=L}
  \operatorname{MSE}(\mPsi_\mathcal{A})$. Then the solution
  $\mathcal{L}$ of FrameSense is near-optimal
  with regards to the $\operatorname{MSE}$,
\begin{align}
&\operatorname{MSE}(\mPsi_\mathcal{L})\le\eta\operatorname{MSE}(\mPsi_\text{OPT})
\text{ with } \eta=\gamma\frac{(d+\delta)^2}{(d-\delta)^2}\frac{L_\text{MAX}}{L_\text{MIN}}, \nonumber
\end{align}
where $\eta$ is the approximation factor of the $\operatorname{MSE}$
and $\gamma$ is the approximation factor of the $\operatorname{FP}$.
\label{thm:MSE}
\end{theorem}
\begin{IEEEproof}
  First, we compute the worst case $\operatorname{MSE}$ when
  FrameSense yields the worst $\operatorname{FP}$, that is for
  $\operatorname{FP}(\mPsi_\calL)=\gamma
  \operatorname{FP}(\mPsi_\text{OPT})$. Using the upper bound
  \eqref{eq:MSEmax} and the bounds on the spectrum for
  $(\delta,L)$-bounded frames, we have
\begin{align}
\operatorname{MSE}(\mPsi_\calL)
\le\frac{K}{L_\text{MIN}}\frac{\gamma \operatorname{FP}(\mPsi_\text{OPT})}{(d-\delta)^2}.
\label{eq:upper_bound}
\end{align}
Then, we compute the best case $\operatorname{MSE}$ when the
$\operatorname{FP}$ is optimal. We note that the lower bound
\eqref{eq:MSEmin} of the $\operatorname{MSE}$ is monotonically
decreasing with regards to the $\operatorname{FP}$. Therefore, we use
the same strategy considered for the lower bound, and obtain
\begin{align}
\operatorname{MSE}(\mPsi_\text{OPT})\le
\frac{K}{L_\text{MAX}}\frac{\operatorname{FP}(\mPsi_\text{OPT})}{(d+\delta)^2}.
\label{eq:lower_bound}
\end{align}
Note that we consider that the optimal $\operatorname{MSE}$ is
achieved for the optimal $\operatorname{FP}$ because the lower bound
of the $\operatorname{MSE}$ is monotonically decreasing with regards
to $\operatorname{FP}.$ Finally, we compute the $\operatorname{MSE}$
approximation ratio as the ratio between \eqref{eq:upper_bound} and
\eqref{eq:lower_bound}, obtaining the desired result.
\end{IEEEproof}

The definition of $(\delta,L)$-bounded frames is key to the proof. It
turns out that many families of adequately normalized random matrices
satisfy Definition \ref{def:bounded_frame}, but it is hard to build
deterministic matrices with such property. Nonetheless, FrameSense
works well even for $\mPsi$ that are not provably
$(\delta,L)$-bounded. The same happens for compressed sensing and RIP
matrices \cite{Monajemi:2013fq}.

\subsection{Practical considerations on FrameSense}
\label{sec:practical}

One point of FrameSense that needs improvement is the optimization of
the sensing energy $L_\mathcal{L}$. In fact, the $\operatorname{FP}$
tends to discard the rows having a larger norm, which in theory could
be more relevant to minimize the $\operatorname{MSE}$. As an example,
consider a matrix $\mPsi$ built as follows,
\begin{align}
\mPsi=\begin{bmatrix}
\mPsi_0\\
C\mPsi_0
\end{bmatrix},
\end{align} where $\mPsi_0$ is a generic matrix and $C>1$ is a
constant. FrameSense would sub-optimally pick rows from the first
matrix, discarding the ones from the second matrix and creating a
sub-optimal sensor placement. In fact, the second matrix would have a
lower $\operatorname{MSE}$ thanks to the multiplicative constant
$C$. To limit such phenomenon, we optimize the sensor location on
$\mPsi'$, that is a $\mPsi$ with unit-norm rows. This solution is not
perfect: it removes the negative bias introduced by the norm of
$\mpsi_i$, but it does not exploit the sensing energy to improve the
sensor placement $\mathcal{L}$. We leave to future work the study of a
new $\operatorname{FP}$-based algorithm able to exploit the
information contained in the norm of the rows.

We conclude this section with a quantification of the bounds given in
Theorem \ref{thm:FP} and Theorem \ref{thm:MSE} in a simple
scenario. Consider a matrix $\mPsi\in\C^{N\times K}$ filled with
i.i.d. Gaussian complex random variables with zero mean and variance
$\sigma^2=1/K$. Assume that $L=c_1K$ and $N=c_2K$, with
$c_2>1>c_1$. Then, we have the following inequalities verified in
expectation,
\begin{align}
&\|\mpsi_i\|_2=1\quad\forall i,\nonumber\\
&\frac{K}{{L_\text{MIN}}^2}=\frac{1}{{c_1}^2K},\nonumber\\
&\operatorname{FP}(\mPsi)={c_2}^2{c_1}^2K.\nonumber
\end{align}
According to Theorem \ref{thm:FP}, the approximation factor of the
$\operatorname{FP}$ is equal to $\gamma=1+\frac{c_2^2-1}{e}$. We
assume that $K$ is sufficiently large and we consider the
Marchenko-Pastur law \cite{Couillet:2013df} to compute the following
approximated bounds for the eigenvalues of $\mPsi_{\calL}$,
\begin{align}
d-\delta=\sqrt{\frac{1}{c_1}}\left(1-\sqrt{c_1}\right)^2\le\lambda_i\le \sqrt{\frac{1}{c_1}}\left(1+\sqrt{c_1}\right)^2=d+\delta.\nonumber
\end{align}
In this scenario, the $\operatorname{MSE}$ approximation factor is equal
to,
\begin{align}
\eta=\left(1+\frac{{c_2}^2-1}{e}\right)\left(\frac{1+\sqrt{c_1}}{1-\sqrt{c_1}}\right)^4.\nonumber 
\end{align}
For example, if $c_1=0.25$ and $c_2=6$, then we have $\gamma\approx14$
and $\eta\approx 50$.

\section{Numerical Results}
\label{sec:num_exp}
In this section, we analyze the performance of FrameSense and compare
it with state-of-the-art algorithms for sensor placement.

\subsection{Synthetic data}
First, we compare the $\operatorname{FP}$ with other cost functions
when used in a naive greedy algorithm. Among the ones listed in Section
\ref{sec:literature}, we select the following three cost functions:
mutual information \cite{Krause:2008vo}, determinant of
$\mT_\mathcal{L}$ \cite{Shamaiah:2010hj}, and $\operatorname{MSE}$
\cite{Das:2008uc}.  We also consider an algorithm that randomly places
the sensors to relate the obtained results to a random selection.
\begin{figure*}[t!]  \centering
  \subfloat{\label{fig:a}\includegraphics[scale=0.45]{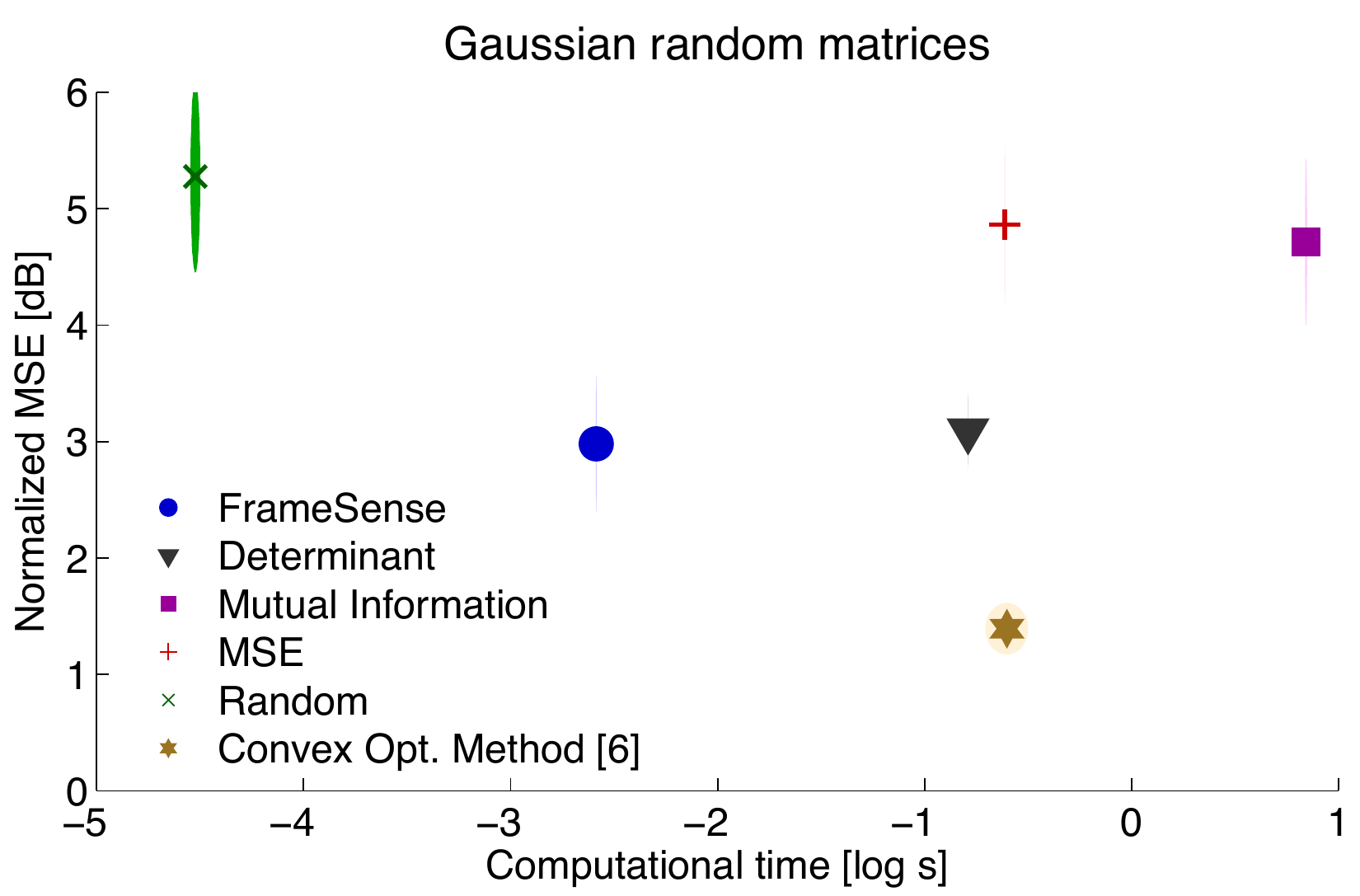}}
  \qquad
  \subfloat{\label{fig:b}\includegraphics[scale=0.45]{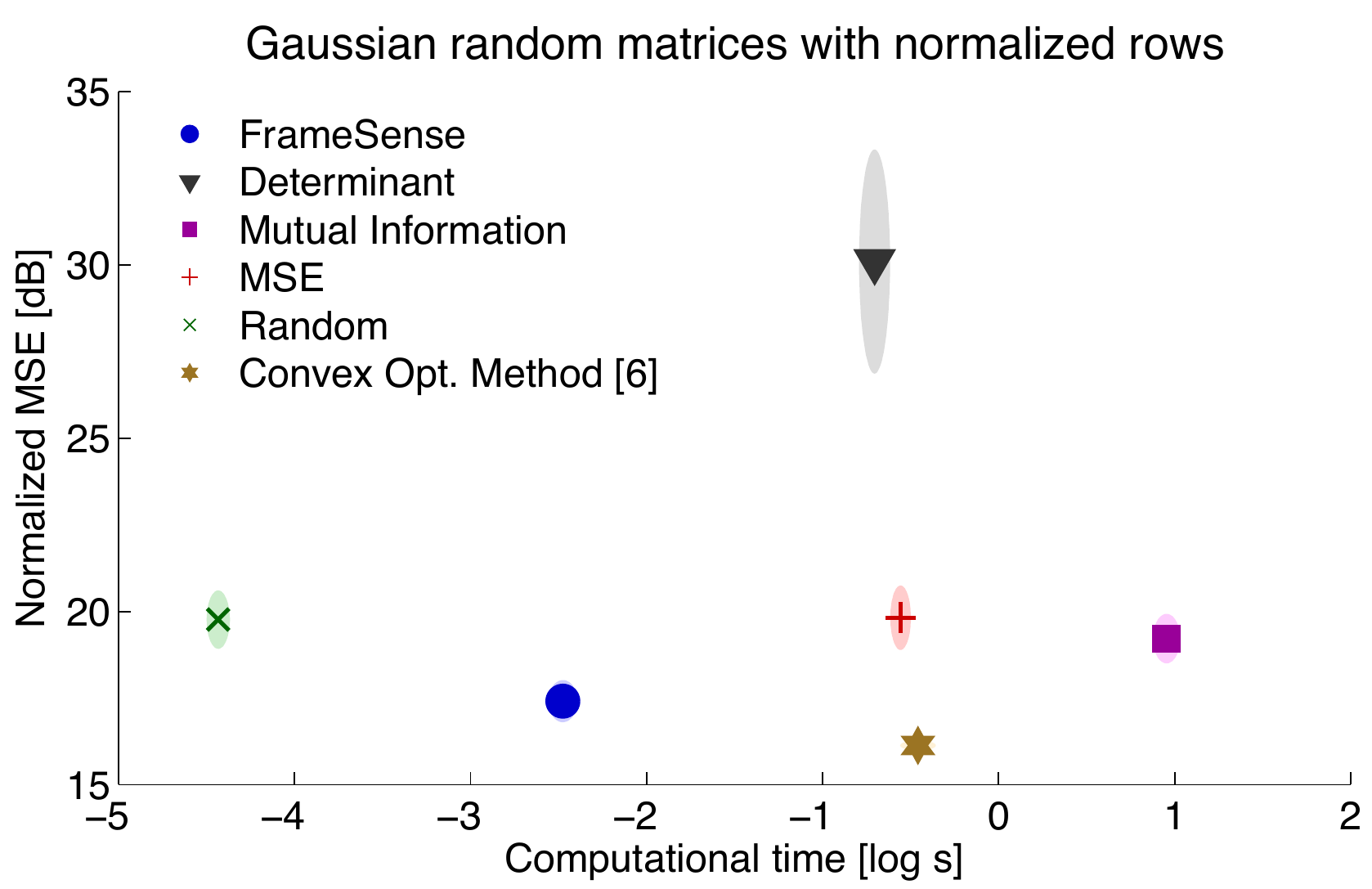}}\\
  \subfloat{\label{fig:c}\includegraphics[scale=0.45]{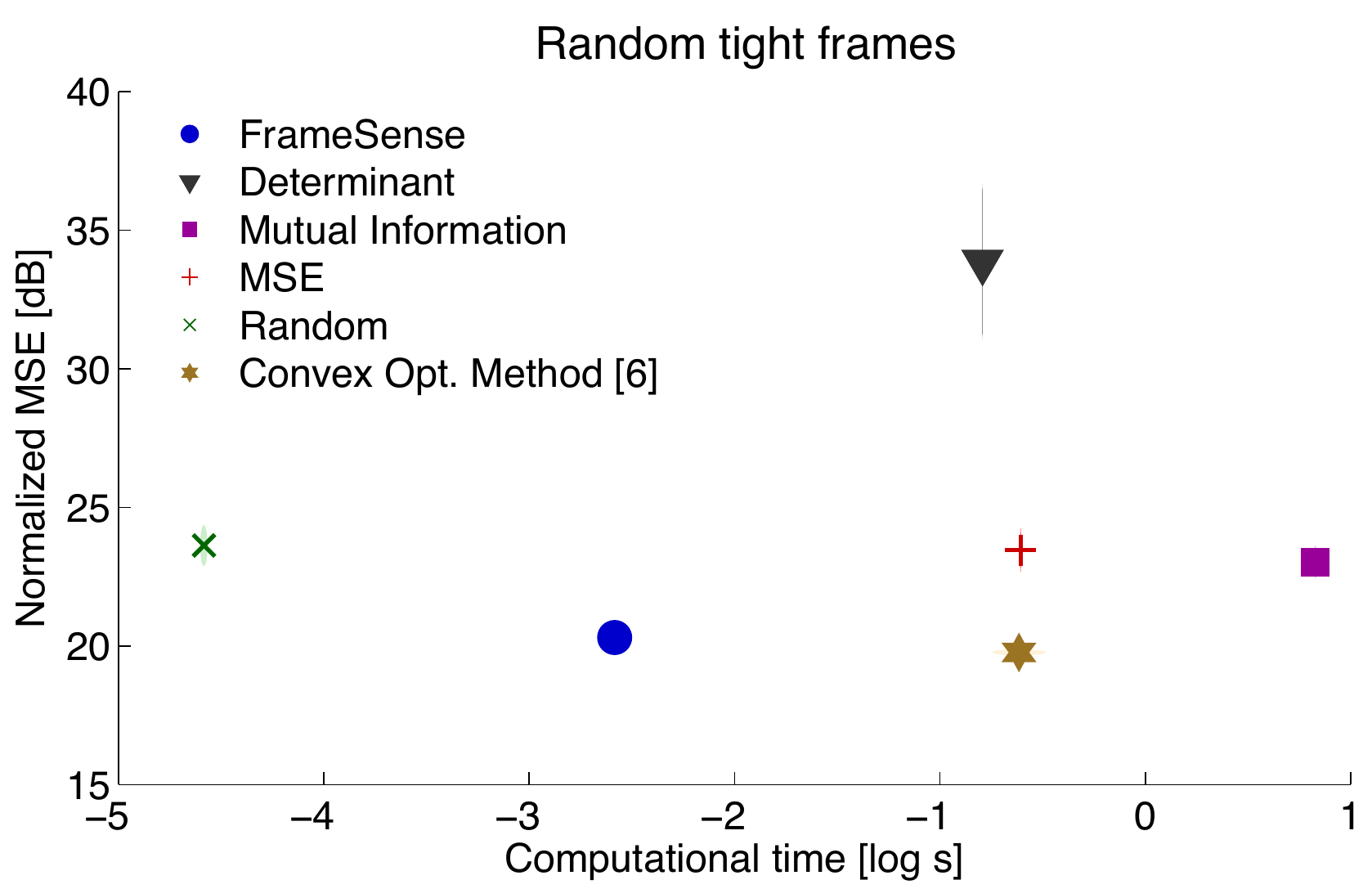}}
\qquad
  \subfloat{\label{fig:b}\includegraphics[scale=0.45]{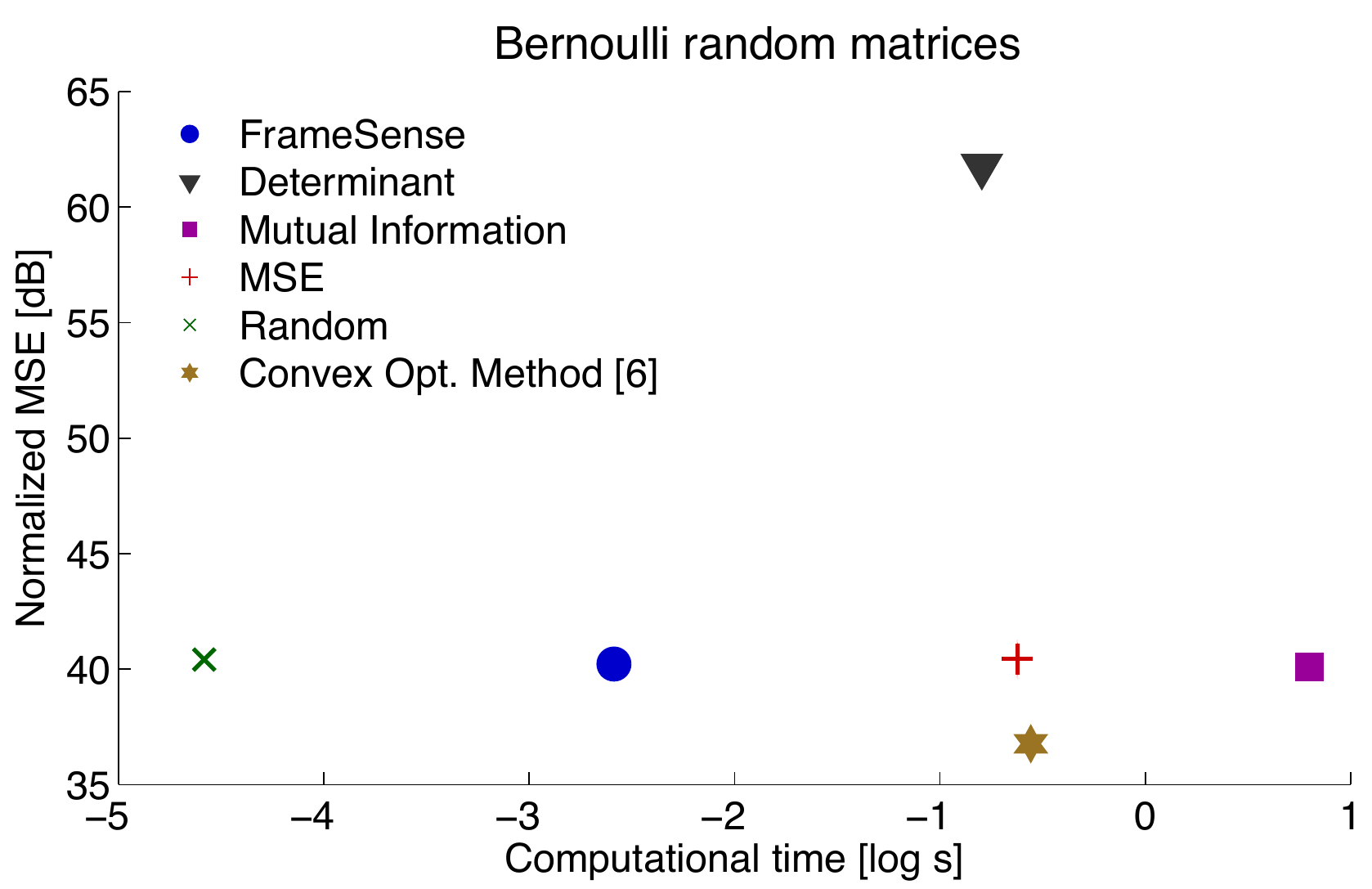}}
\vspace{-0.3cm}
\caption{Tradeoff between computational time and
  $\operatorname{MSE}$. We randomly generate matrices---according to
  four different models---with $N=100$, $K=30$ and $L=50$ sensors. The
  performance is measured in terms of $\operatorname{MSE}$, therefore
  the lower the dot, the higher the performance. On the other hand,
  the computational time is measured in seconds.  The error bars are
  represented as ellipsoids, where the length of the axes represent
  the standard deviation of the data point. Note that, FrameSense is
  the fastest algorithm by one order of magnitude and it is the second
  best algorithm in terms of $\operatorname{MSE}$. When the legend
  refer to a cost function, we considered a naive greedy algorithm
  optimizing that cost function. On the other hand, we refer to the
  authors for specific algorithms implementing complicated schemes
  and/or heuristics. }
  \label{fig:comp_time}
\vspace{-0.4cm}
\end{figure*}

\begin{figure}[t!]
  \centering 
  \includegraphics[scale=0.45]{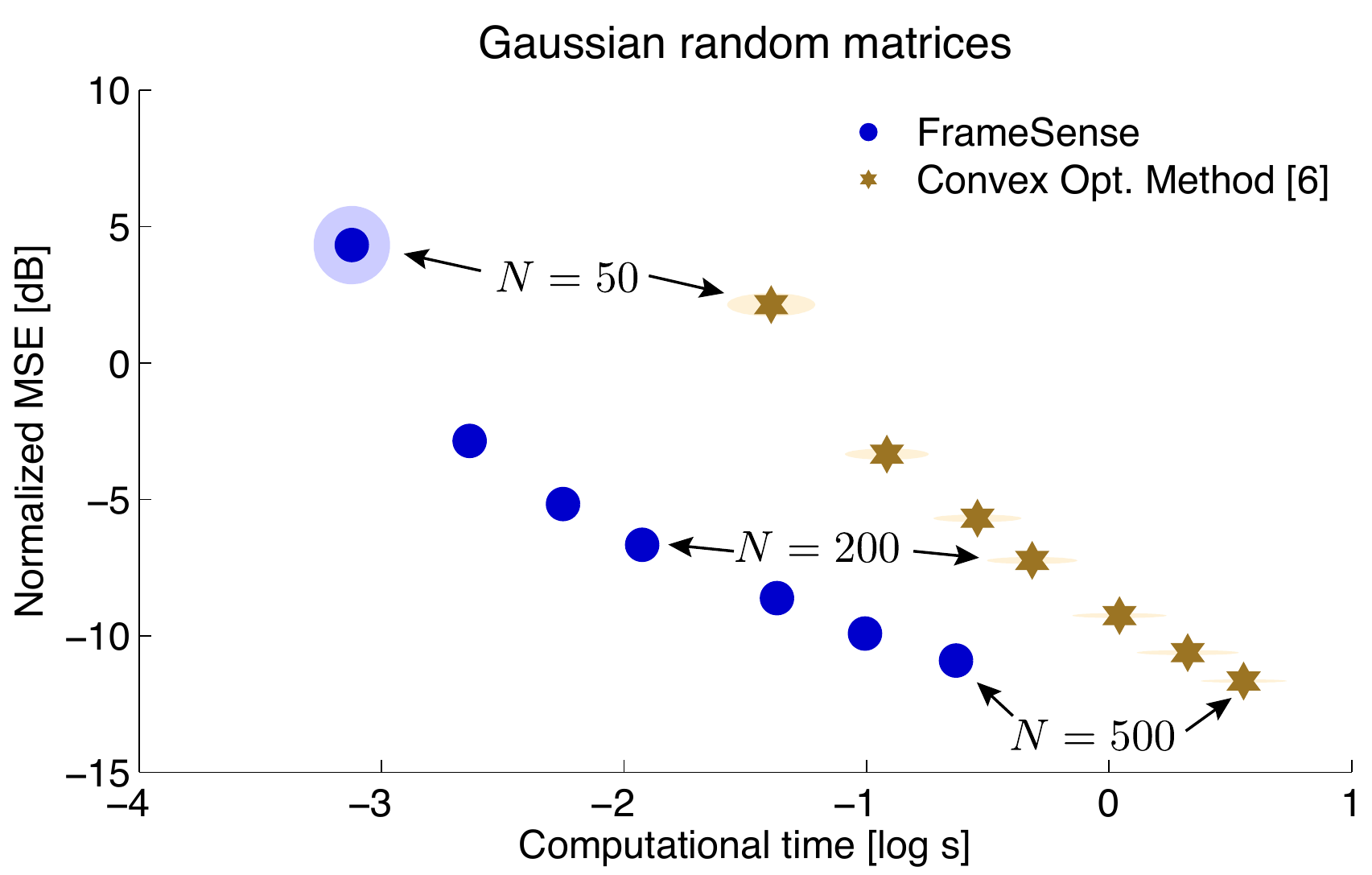}
\vspace{-0.3cm}
\caption{Analysis of the tradeoff between computational time and
  $\operatorname{MSE}$ for FrameSense and the convex relaxed algorithm
  proposed by Joshi et al. \cite{Joshi:2009el}. We generate 100
  Gaussian matrices with $K=20$ and of increasing size
  $N=\{50,100,150,200,300,400,500\}$, while we place $L=0.5N$
  sensors. The error bars are represented as ellipsoids, where the
  length of the axes represent the standard deviation of the data
  point.  We measure the average computational time together with the
  average $\operatorname{MSE}$, showing that while FrameSense is
  significantly faster than the convex algorithm, the difference in
  $\operatorname{MSE}$ is minimal. Moreover, the gap in the quality of
  the solution decreases for an increasing size of the problem $N$.  }
  \label{fig:comp_joshi}
\vspace{-0.5cm}
\end{figure}

The greedy algorithms are tested on different types of sensing
matrices $\mPsi$:
\begin{itemize}
\item random matrices with Gaussian i.i.d. entries,
\item random matrices with Gaussian i.i.d. entries whose rows are
  normalized,
\item random matrices with Gaussian i.i.d. entries with
  ortho-normalized columns, that we call random tight frame due to the
  Naimark theorem \cite{Kovacevic:2007fj},
\item random matrices with Bernoulli i.i.d entries.
\end{itemize}
Note that the use of random matrices is sub-optimal, since we would
rarely encounter such a case in a real-world scenario. However, it is
the only available \emph{dataset} that allows us to test thoroughly the different
algorithms.

We consider $\mPsi\in\R^{100\times30}$ and evaluate the performance in
terms of $\operatorname{MSE}$ for $L=\{35,35,40,45,50,55,60\}$. We use
100 different instances for each combination, and we compute the
average $\operatorname{MSE}$ as a function of $L$. Note that the
$\operatorname{MSE}$ is always computed using \eqref{eq:MSE}, which
assumes i.i.d. Gaussian noise perturbing the measurements $\vf$ and a
uniform distribution on $\vec{\alpha}$. The relatively small size of
$\mPsi$ and the low number of trials are due to the lack of
scalability of certain cost functions, which require the computation
of large matrices. The results are given in Figure
\ref{fig:greedy_comp}. We note that FrameSense is consistently
outperforming all other cost functions. In the random Gaussian
matrices case, the determinant shows similar results. However, looking
at the Bernoulli matrices case, we see that the determinant leads to a
significantly worse $\operatorname{MSE}$. Note that certain cost
functions show worse performance than a random selection of the
rows. While this phenomena could be partially explained by the special
properties of certain families of random matrices, it indicates the
importance of choosing a well-studied cost function for which we can
obtain performance bounds with respect to the $\operatorname{MSE}$.

According to the theory of random matrices, any selection of rows
should have the same spectrum on expectation. Therefore, it should not
be possible to outperform the random selection of rows. However, the
theoretical analysis of random matrices is generally valid only
asymptotically, while here we show results for relatively small
matrices. This phenomenon could also be of interest for other domains,
such as compressed sensing, indicating the possibility of
\emph{optimizing} the spectrum of random matrices to improve their
performance.

In Figure \ref{fig:greedy_comp_time}, we show the average
computational time with regards to the number of possible sensor
locations $N=\{20,50,80,\ldots,200\}$. We consider 100 random Gaussian
matrices $\mPsi\in\R^{N\times 10}$ and we place $0.5N$ sensors.  We
underline that FrameSense is significantly faster than any other
greedy algorithm, the only exception being the random selection that
has a computational time close to zero. Note that the other
parameters, such as as $L$ and $K$, have little influence on the
computational time, which strongly depends on $N$.

In a second experiment, we compare FrameSense with a state-of-the-art
method based on convex optimization \cite{Joshi:2009el}. Since the
algorithm proposed by Joshi et al. \cite{Joshi:2009el} is structurally
different from FrameSense, we focus this analysis on two parameters:
the computational time and the $\operatorname{MSE}$. We fix $K=20$ and
the ratio between the number of sensors and the number of available
locations as $L/N=0.5$. Then, we vary the number of possible locations
as $N=\{50,100,150,200,300,400,500\}$. The results for Gaussian random
matrices are given in Figure \ref{fig:comp_joshi}. First, we note that
the convex method achieves a lower $\operatorname{MSE}$, however the
performance gap decreases when we increase the number of sensors. This
is not surprising, since FrameSense is a greedy algorithm that does
not require parameter fine-tuning nor heuristics, while the convex
relaxation method integrates some efficient heuristics. For example,
at every iteration, it refines the selection by looking at all the
possible $\emph{swaps}$ between the chosen locations and the discarded
ones. This strategy is particularly effective when $L\approx K$: in
fact, swapping just one row can improve significantly the spectrum of
$\mT_\calL$, and consequently the $\operatorname{MSE}$ achieved by
$\mPsi_\mathcal{L}$. The heuristics, while effective in terms of
$\operatorname{MSE}$, increase also the computational cost. In fact,
FrameSense is significantly faster.

The last comparison also opens an interesting direction for future
work. In fact, the convex relaxed algorithm optimizes the determinant
of $\mT_\mathcal{L}$ and its near-optimality in terms of
$\operatorname{MSE}$ has not been shown. Moreover, this cost function
has been proven to be less effective compared to $\operatorname{FP}$
when used in a greedy algorithm (see Figure
\ref{fig:greedy_comp}). Therefore, we expect that a convex relaxed
scheme based on the $\operatorname{FP}$ has the potential to define a
new state of the art, mixing the advantages of $\operatorname{FP}$ and
the heuristics proposed in \cite{Joshi:2009el}.

To conclude the performance analysis, we study the trade-off between
computational complexity and performance for all the considered
algorithms, greedy and not. We picked 100 instances of each of the
random matrices proposed in the first experiment with $N=100$, $L=40$
and $K=30$. We measured the average computational time and average
$\operatorname{MSE}$ obtained by each algorithm and the results are
given in Figure \ref{fig:comp_time}. We note a general trend
connecting the four subfigures: FrameSense is the fastest algorithm,
by at least an order of magnitude, while its performance is just
second, as previously shown, to the convex relaxed method proposed by
Joshi et al. \cite{Joshi:2009el}.

Since the sensor placement is an off-line procedure, we may argue that
the computational time is of secondary importance. While this is true
in many applications, there are certain applications where it is
necessary to recompute $\mathcal{L}$ regularly. This is usually the
case when $\mPsi$ changes in time due to changes of the physical field
and it is possible to adaptively reallocate the sensors. In other
applications, such as the ones where we attempt to interpolate the
entire field from $L$ measurements, the number of possible locations
$N$ grows with the desired resolution. In this case, a lower
computational time is of critical importance.

\begin{figure*}[t!]
  \centering 
  \subfloat{\label{fig:tmaps_a}\includegraphics[scale=0.45]{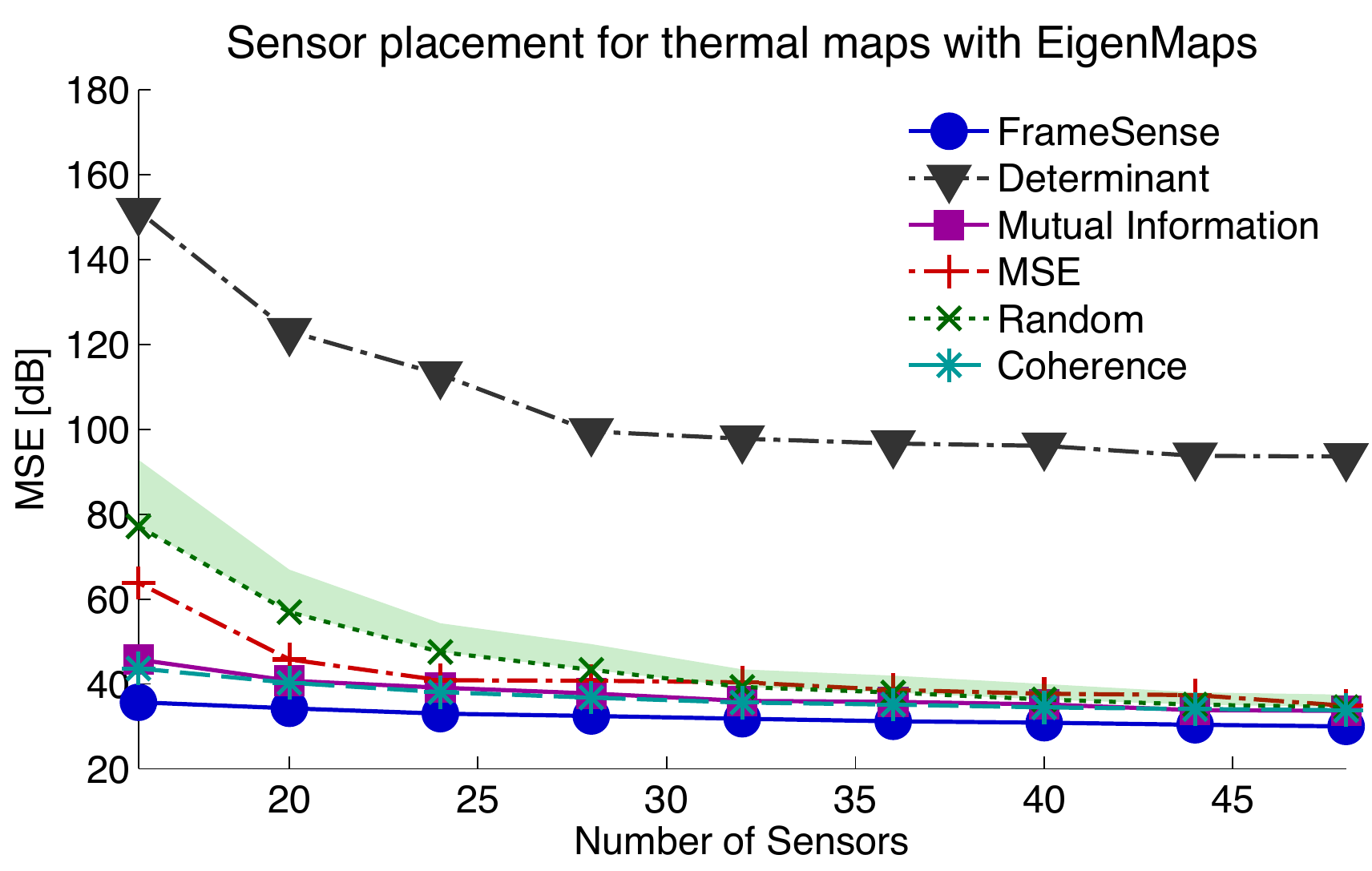}}
  \qquad
  \subfloat{\label{fig:tmaps_b}\includegraphics[scale=0.45]{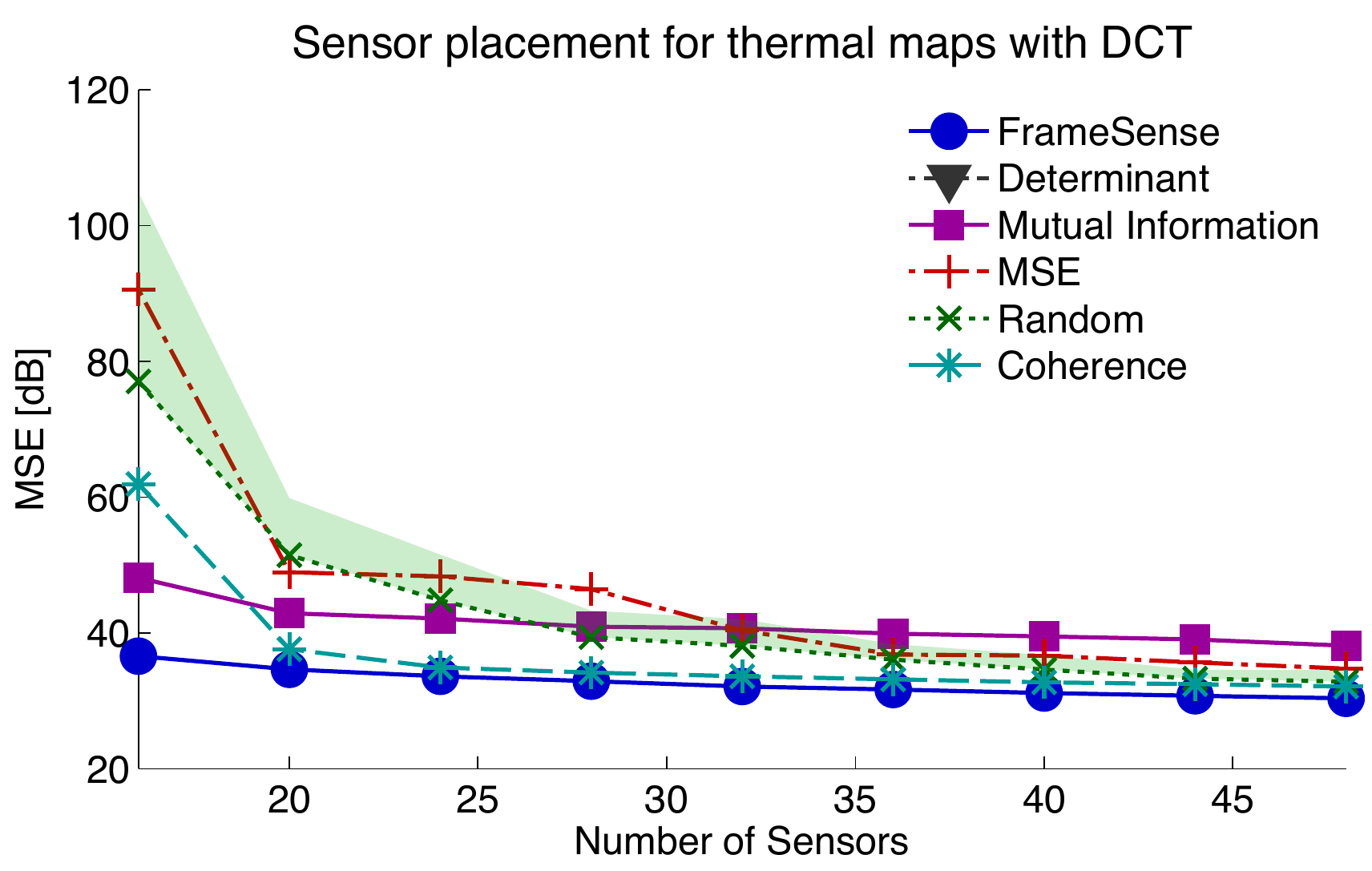}}
\vspace{-0.3cm}
\caption{Comparison between FrameSense, the other greedy algorithms
  and a coherence-based greedy algorithm proposed in
  \cite{Ranieri:2012fh}. In this experiment, we consider the sensors
  placement to estimate the temperature of an 8-core microprocessor
  using a limited number of sensors. Two different matrices are
  proposed: on the left, the matrix $\Psi\in\R^{420\times 16}$ is
  generated from a principal component analysis of known thermal maps
  as in \cite{Ranieri:2012fh}; on the right, the matrix
  $\Psi\in\R^{420\times 16}$ is the subsampled DCT matrix proposed in
  \cite{Nowroz:2010bt}. The shaded area represents the positive side
  of the error bar for the random sensor placement, measured using the
  standard deviation over 100 realizations. Note that FrameSense
  significantly outperforms the previous coherence-based method and
  the other greedy algorithms, in particular when the number of
  sensors is close to the number of estimated parameters, that is
  $K=16$. }
  \label{fig:tmaps}
\vspace{-0.3cm}
\end{figure*}

\subsection{Temperature estimation on many-core processors}
We now analyze the impact of FrameSense on a real-world problem where
sensor placement is of fundamental importance. We describe the
problem, followed up by a simulation showing the improvement with
respect to the state of the art.

The continuous evolution of process technology increases the
performance of processors by including more cores memories and complex
interconnection fabrics on a single chip. Unfortunately, a higher
density of components increases the power densities and amplifies the
concerns for thermal issues. In particular, it is key to design
many-core processors that prevent hot spots and large on-chip
temperature gradients, as both conditions severely affect the overall
system's characteristics. A non-exhaustive list of problems induced by
thermal stress includes higher failure rate, reduced performance,
increased power consumption due to current leakage and increased
cooling costs.  To overcome these issues, the latest designs include
the thermal information into the workload allocation strategy to
obtain optimal performance while avoiding critical thermal
scenarios. Consequently, a few sensors are deployed on the chip to
collect thermal data. However, their number is
limited by area/power constraints and their optimal placement, that
detects all the worst-case thermal scenarios, is still unresolved and
has received significant attention
\cite{Cochran:2009ba,Mukherjee:2006joa,Nowroz:2010bt,Reda:kj,Sharifi:ip}.

An improved sensor placement algorithm would lead to a reduction of
the sensing cost in terms of used silicon surface. Moreover, it
implies a reduction of the reconstruction error, making possible the
use of more aggressive scheduling strategies and consequently improves
the processor performance. In \cite{Ranieri:2012fh}, we proposed the
following strategy: learn $\mPsi$ using a principal component analysis
on an extensive set of simulated thermal maps, then place the sensors
with a greedy algorithm minimizing the coherence between the
rows. This work achieved a significant improvement compared to the
previous state-of-the-art approach \cite{Nowroz:2010bt} that was based
on specific knowledge of the spectrum of the thermal maps.

We consider a sensing matrix $\mPsi$ equivalent to one defined in
\cite{Ranieri:2012fh} but at a lower resolution for computational
reasons. More precisely, we consider thermal maps at a resolution of
$16\times 15$ and a matrix $\mPsi\in\{420\times 16\}$.  Then, we
compare FrameSense, our previous greedy algorithm from
\cite{Ranieri:2012fh} and the other algorithms considered in Figure
\ref{fig:greedy_comp}. The results are shown in Figure
\ref{fig:tmaps}. We note that the performance of the placement
algorithm has been further improved by FrameSense, without increasing
the computational cost or changing the reconstruction
strategy. Moreover, we are now able to guarantee the near-optimality
of the algorithm with respect to the $\operatorname{MSE}$ of the
estimated thermal map.

Note that we are only discussing the impact of an optimized sensor
placement when we consider the $\mPsi$ based on the Eigenmaps
described in \cite{Ranieri:2012fh}.  The joint problem of sensor
placement and reconstruction of thermal maps is more complex and other
factors may play a fundamental role. For example, it may be more
convenient to use an opportunely constructed DCT frame
\cite{Nowroz:2010bt} to reduce the memory occupation at the price of a
reduced reconstruction precision. The comparison of the reconstruction
performance in terms of $\operatorname{MSE}$ for the various
algorithms when considering the $\mPsi\in\R^{420\times16}$ based on
the DCT frame are given in Figure \ref{fig:tmaps}. Again, FrameSense
outperforms all the other placement algorithms, in particular when the
number of sensors is limited.

In Figure \ref{fig:tmaps2}, we present an example of a thermal map and
the two reconstructions obtained by sensor placements optimized by
either FrameSense or the coherence-based algorithm. Notably, the
reconstruction obtained with the sensor placement proposed by
FrameSense is significantly more precise.

\section{Conclusions}

We studied the optimization of sensor placement when the collected
measurements are used to solve a linear inverse problem. The problem
is equivalent to choosing $L$ rows out of $N$ from a matrix $\mPsi$
such that the resulting matrix has favorable spectral properties. The
problem is intrinsically combinatorial and approximation algorithms
are necessary for real-world scenarios. While many algorithms have
been proposed, none has guaranteed performance in terms of the
$\operatorname{MSE}$ of the solution of the inverse problem, which is
the key merit figure.

We proposed FrameSense, a greedy worst-out algorithm minimizing the
$\operatorname{FP}$. Even if this chosen cost function is well-known in
frame theory for its fundamental role in the construction of frames
with optimal $\operatorname{MSE}$, FrameSense is the first algorithm
exploiting it as a cost function for the sensor placement problem. Our
theoretical analysis demonstrates the following innovative aspects:
\begin{itemize}
\item FrameSense is near-optimal with respect to the
  $\operatorname{FP}$, meaning that it always places the sensors such
  that the obtained $\operatorname{FP}$ is guaranteed to be close to the
  optimal one.
\item under RIP-like assumptions for $\mPsi$, FrameSense is also
  near-optimal with respect to the $\operatorname{MSE}$. Note that
  FrameSense is the first algorithm with this important property.

\end{itemize}

\begin{figure}[t!]
  \centering 
  \includegraphics[scale=0.35]{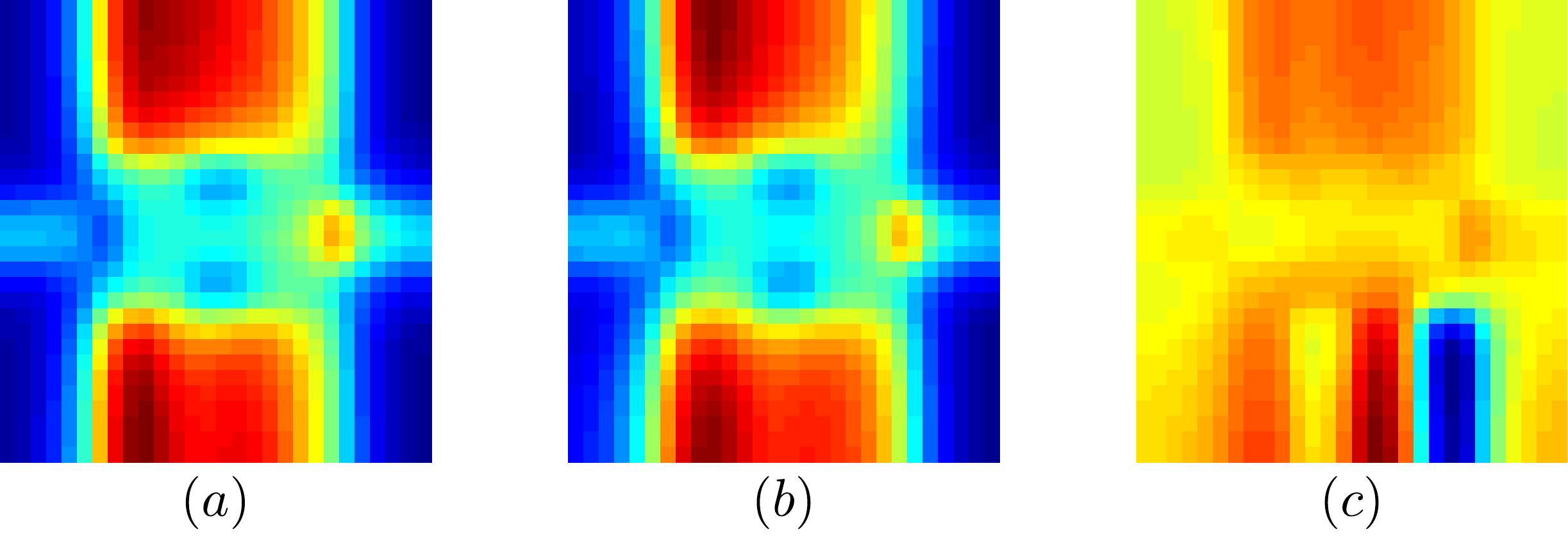}
\vspace{-0.3cm}
\caption{Examples of reconstruction of a thermal map of resolution
  $56\times 60$ using two different sensor placements. The real
  thermal map is depicted in $(a)$, while the reconstructed thermal
  maps using FrameSense and the algorithm proposed in
  \cite{Ranieri:2012fh} are given in $(b)$ and $(c)$,
  respectively. $L=16$ sensors are placed to measure $K=16$ parameters
  and the measurements are corrupted by noise with an SNR of 8.5
  dB. Note that the reconstruction using FrameSense is close to the
  real map, while the other one is extremely noisy. The details
  regarding the generation of the thermal maps are available in
  \cite{Ranieri:2012fh}.  }
  \label{fig:tmaps2}
\vspace{-0.6cm}
\end{figure}

We provided extensive numerical experiments showing that FrameSense
achieves the best performance in terms of $\operatorname{MSE}$ while
having the lowest computational complexity when compared to many other
greedy algorithms. FrameSense is also competitive performance-wise
with a state-of-the-art algorithm based on a convex relaxation
proposed in \cite{Joshi:2009el}, while having a substantially smaller
computational time.

We showed that FrameSense has appealing performance on a real-world
scenario, the reconstruction of thermal maps of many-core
processors. We showed a potential for reducing the number of sensors
required to estimate precisely the thermal distribution on a chip,
reducing the occupied area and the consumed power by the sensors.

Future work will be three-fold. First, it is foreseeable to relax the
RIP-like condition on $\mPsi$ by considering that the characteristics
of the $\operatorname{FP}$ are potentially sufficient to avoid
$\mPsi_\calL$ matrices with an unfavorable spectral
distribution. Moreover, it would be interesting to show that there
exists matrices, random or deterministic, that are
$(\delta,L)$-bounded frames. Second, we believe that a convex relaxed
scheme based on the $\operatorname{FP}$ integrating the heuristics
proposed by Joshi et al. \cite{Joshi:2009el} could improve
significantly the $\operatorname{MSE}$ of the obtained solution, while
keeping the near-optimality thanks to the $\operatorname{FP}$. Third,
we would like to derive a new cost function that exploits the sensing
energy of the rows, as described in Section \ref{sec:practical}.

\appendix

\subsection{Proof of Lemma \ref{lemma:MSE}}
\label{app:bound_MSE}
In this section, we bound the $\operatorname{MSE}$ of a matrix
$\mPsi_\mathcal{A}\in\R^{L\times K}$ as a
function of its $\operatorname{FP}$ and the spectrum
$\set{\lambda_i}_{i=1}^K$ of $\mT_\calL$.  

First, consider the harmonic mean $\operatorname{H}=\frac{K}{\sum_k
  \frac{1}{\lambda_j}}$, the arithmetic mean
$\operatorname{A}=\frac{\sum_k\lambda_k}{K}$ and the standard
deviation $\operatorname{S}=\sqrt{\frac{1}{K}\sum_k\left(\lambda_k -
    A\right)^2}$ of the eigenvalues of $\mT_\calL$. All these
quantities are linked to
$\operatorname{MSE}\left(\mPsi_\mathcal{A}\right)$,the number of
sensors $L$ and $\operatorname{FP}
\left(\mPsi_\mathcal{A}\right)$. More precisely, we have
\begin{align}
  &\operatorname{H}=\frac{K}{\operatorname{MSE}\left(\mPsi_\mathcal{A}\right)},\nonumber\\
  &\operatorname{A}=\frac{L_\mathcal{A}}{K},\nonumber\\
  &\operatorname{S}=\sqrt{\frac{1}{K}\left(\operatorname{FP}\left(\mPsi_\mathcal{A}\right)-\frac{{L_\mathcal{A}}^2}{K}\right)}.\nonumber
\end{align}
Then, we consider the following bounds for the harmonic mean of a set
of positive numbers derived by Sharma \cite{Sharma:2008ur},
\begin{align}
  \frac{(M-\operatorname{S})^2}{M(M-2\operatorname{S})}\le\frac{\operatorname{A}}{\operatorname{H}}\le
  \frac{(m+\operatorname{S})^2}{m(m+2\operatorname{S})},\nonumber
\end{align}
where $m$ and $M$ are the smallest and the largest number in the
set. We use the expressions of $\operatorname{A}$ and
$\operatorname{H}$ and we remove the mixed term in the denominator to
obtain,
\begin{align}
  \frac{K^2}{L_\mathcal{A}}\left(1+\frac{\operatorname{S}^2}{M^2}\right)\le
\operatorname{MSE}\left(\mPsi_\mathcal{A}\right)\le
  \frac{K^2}{L_\mathcal{A}}\left(1+\frac{\operatorname{S}^2}{m^2}\right).\nonumber
\end{align}
As expected when the $\operatorname{FP}$ achieves its global minima, that is
$\operatorname{S}=0$, we achieve the optimal MSE of a tight frame.  

To conclude the proof, we consider the two bounds separately starting
from the lower one. Let $M=\lambda_1$ and we plug in the value of $S$.
We also consider without loss of generality $L_\mathcal{A}\le L_\text{MAX}$, since we can
always improve the $\operatorname{MSE}$ by increasing the sensing
power $L_\mathcal{A}$. Then,
\begin{align}
\operatorname{MSE}\left(\mPsi_\mathcal{L}\right)\ge\frac{K^2}{L_\text{MAX}}\left(1+\frac{\operatorname{FP}(\mPsi_\mathcal{A})}{K{\lambda_1}^2}-\frac{{L_\text{MAX}}^2}{K^2{\lambda_1}^2}\right). \nonumber
\end{align}
We obtain the final result by using lower bound on the
largest eigenvalue: $\lambda_1\ge \frac{L_\text{MAX}}{K}$. 

The approach to prove the upper bound is exactly
symmetrical. Specifically, consider $m=\lambda_N$, $L_\mathcal{A}\ge
L_\text{MIN}$ and use the upper bound on the smallest eigenvalue
$\lambda_N\le \frac{L_\text{MAX}}{K}$.

\section*{Acknowledgment}
The authors would like to thank Ivan Dokmani{\'c} and Prof. Ola Svensson
whose suggestions were fundamental to strengthen the results described
in the paper.  This research was supported by an ERC Advanced Grant --
Support for Frontier Research -- SPARSAM Nr: 247006.

\begin{IEEEbiography}[{\includegraphics[width=1in,height=1.25in,clip,keepaspectratio]{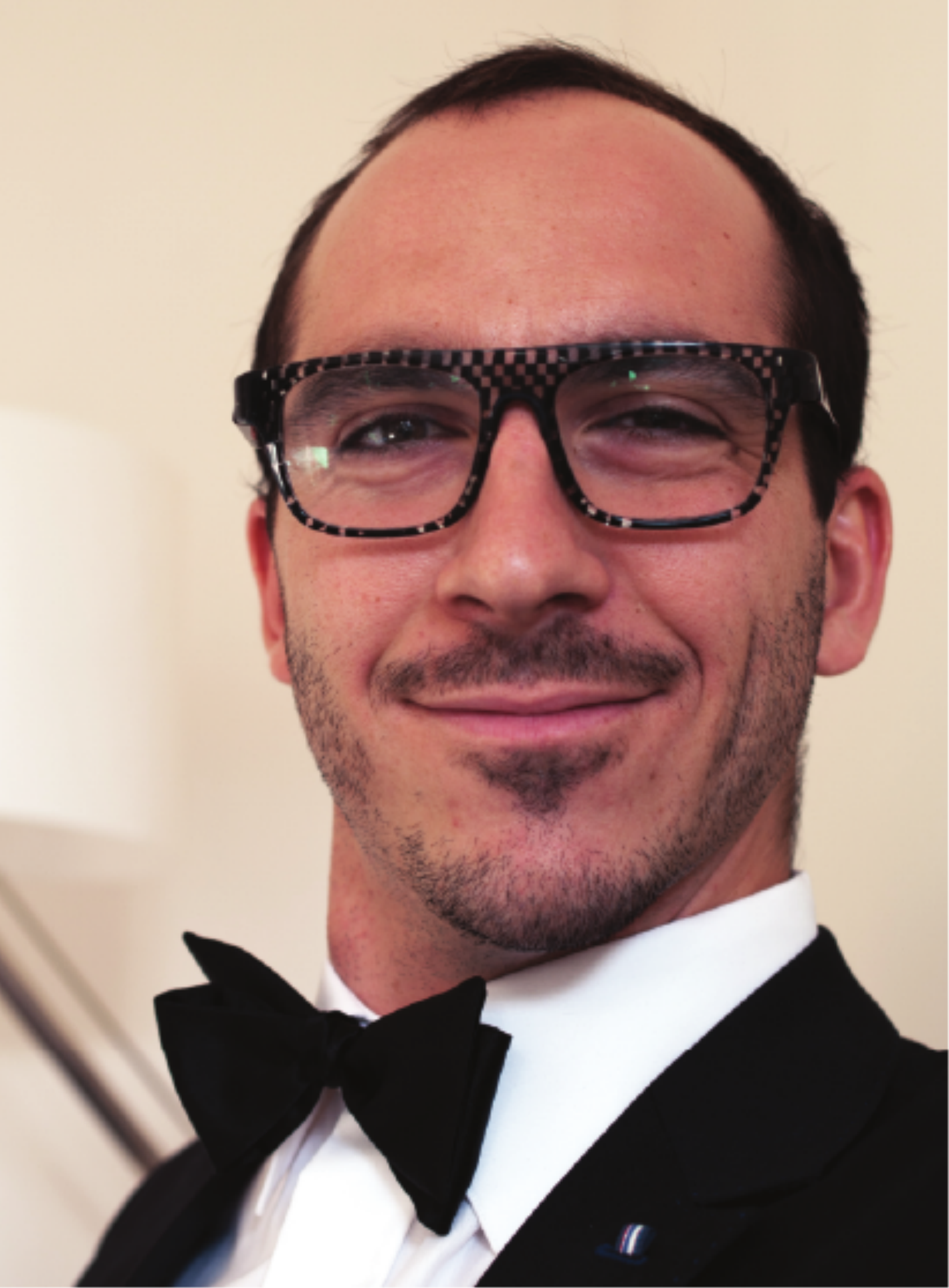}}]{Juri
    Ranieri} received both his M.S. and B.S. degree in Electronic
  Engineering in 2009 and 2007, respectively, from Universit{\'a} di
  Bologna, Italy. From July to December 2009, he joined as a visiting
  student the Audiovisual Communications Laboratory (LCAV) at EPF
  Lausanne, Switzerland. From January 2010 to August 2010, he was with
  IBM Zurich to investigate the lithographic process as a signal
  processing problem.  From September 2010, he is in the doctoral
  school at EPFL where he joined LCAV under the supervision of
  Prof. Martin Vetterli and Dr. Amina Chebira. From April 2013 to July
  2013, he was an intern at Lyric Labs of Analog Devices, Cambridge,
  USA. His main research interests are inverse problems of physical fields
  and the spectral factorization of autocorrelation functions.
\end{IEEEbiography}
\begin{IEEEbiography}[{\includegraphics[width=1in,height=1.25in,clip,keepaspectratio]{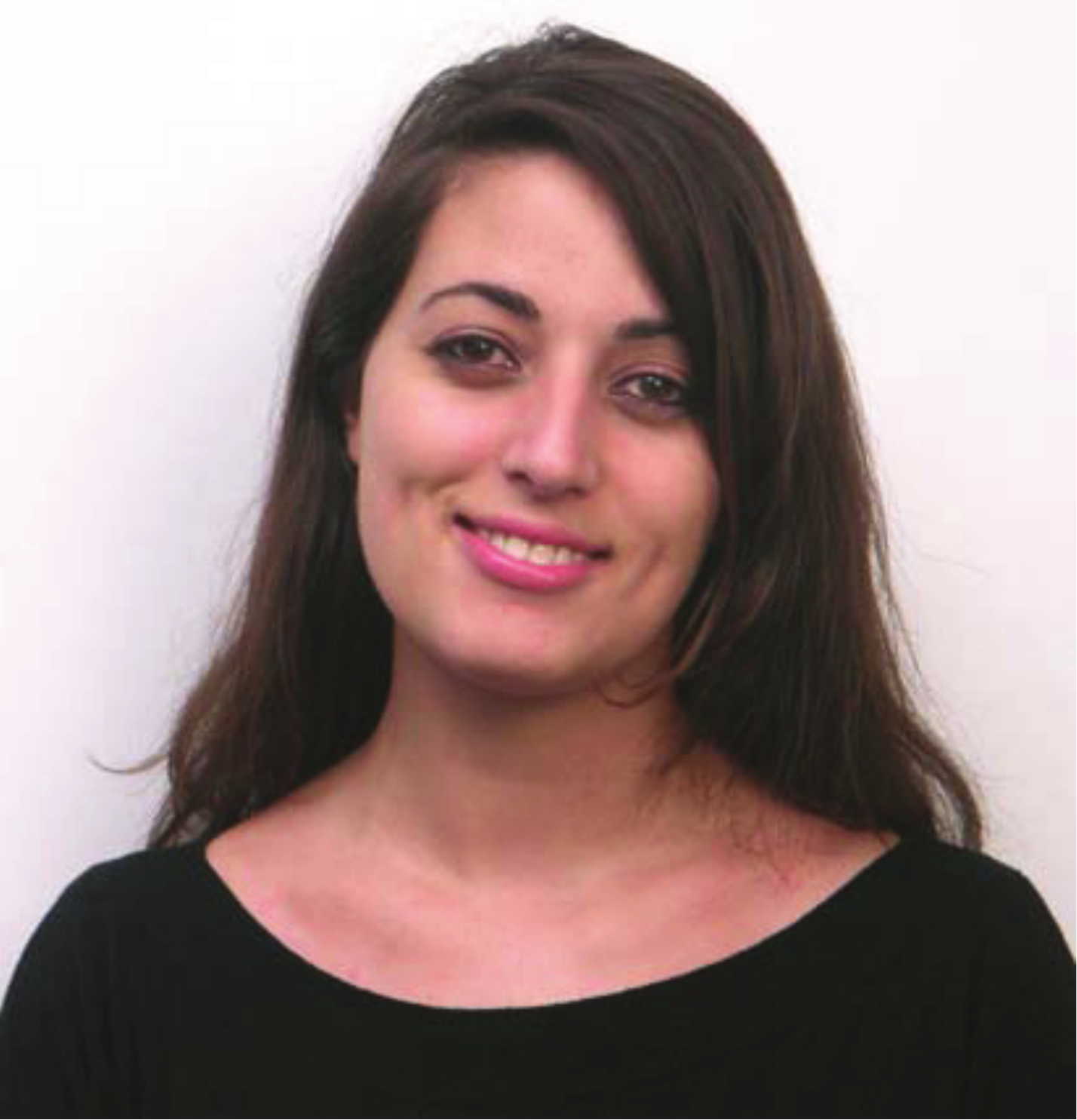}}]{Amina
    Chebira} is a senior research and development engineer at the
  Swiss Center for Electronics and Microtechnology (CSEM) in Neuch{\^
    a}tel, Switzerland. In 1998, she obtained a Bachelor degree in
  mathematics from University Paris 7 Denis Diderot. She received the
  B.S. and M.S. degrees in communication systems from the Ecole
  Polytechnique Fédérale de Lausanne (EPFL) in 2003 and the
  Ph.D. degree from the Biomedical Engineering Department, Carnegie
  Mellon University, Pittsburgh, PA, in 2008, for which she received
  the biomedical engineering research award. She then held a
  Postdoctoral Researcher position with the Audiovisual Communications
  Laboratory, EPFL, from 2008 to 2012. Her research interests include
  frame theory and design, biomedical signal and image processing, pattern
  recognition, filterbanks and multiresolution theory. 
\end{IEEEbiography}

\begin{IEEEbiography}[{\includegraphics[width=1in,height=1.25in,clip,keepaspectratio]{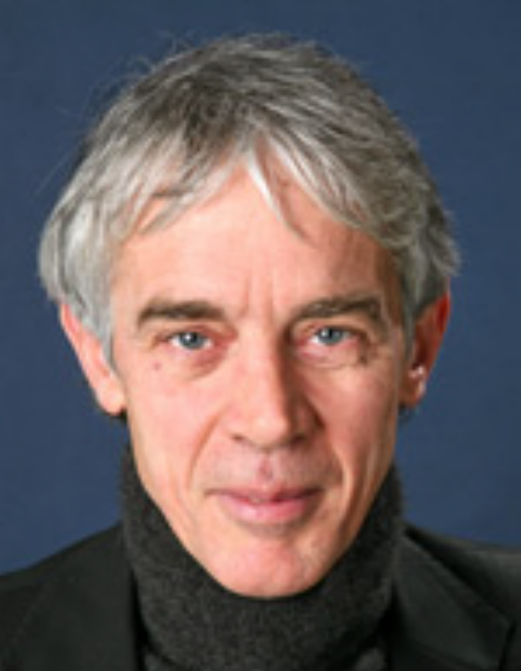}}]{Martin
    Vetterli} was born in 1957 and grew up near Neuchatel. He received
  the Dipl. El.-Ing. degree from Eidgenossische Technische Hochschule
  (ETHZ), Zurich, in 1981, the Master of Science degree from Stanford
  University in 1982, and the Doctorat es Sciences degree from the
  Ecole Polytechnique Federale, Lausanne, in 1986. After his
  dissertation, he was an Assistant and then Associate Professor in
  Electrical Engineering at Columbia University in New York, and in
  1993, he became an Associate and then Full Professor at the
  Department of Electrical Engineering and Computer Sciences at the
  University of California at Berkeley. In 1995, he joined the EPFL as
  a Full Professor. He held several positions at EPFL, including Chair
  of Communication Systems and founding director of the National
  Competence Center in Research on Mobile Information and
  Communication systems (NCCR-MICS). From 2004 to 2011 he was Vice
  President of EPFL and from March 2011 to December 2012, he was the
  Dean of the School of Computer and Communications Sciences. Since
  January 2013, he leads the Swiss National Science Foundation. He
  works in the areas of electrical engineering, computer sciences and
  applied mathematics. His work covers wavelet theory and
  applications, image and video compression, self-organized
  communications systems and sensor networks, as well as fast
  algorithms, and has led to about 150 journals papers. He is the
  co-author of three textbooks, with J. Kovacevic, ”Wavelets and
  Subband Coding” (PrenticeHall, 1995), with P. Prandoni, ”Signal
  Processing for Communications”, (CRC Press, 2008) and with
  J. Kovacevic and V. Goyal, of the forthcoming book ”Fourier and
  Wavelet Signal Processing” (2012). His research resulted also in
  about two dozen patents that led to technology transfers to
  high-tech companies and the creation of several start-ups. His work
  won him numerous prizes, like best paper awards from EURASIP in 1984
  and of the IEEE Signal Processing Society in 1991, 1996 and 2006,
  the Swiss National Latsis Prize in 1996, the SPIE Presidential award
  in 1999, the IEEE Signal Processing Technical Achievement Award in
  2001 and the IEEE Signal Processing Society Award in 2010. He is a
  Fellow of IEEE, of ACM and EURASIP, was a member of the Swiss
  Council on Science and Technology (2000-2004), and is an ISI highly
  cited researcher in engineering.
\end{IEEEbiography}


\bibliographystyle{IEEEbib2}

\bibliography{biblio.bib}

\end{document}